\documentclass[a4paper,UKenglish,cleveref, autoref, thm-restate]{lipics-v2021}

\hideLIPIcs  
\nolinenumbers

\usepackage{macros}
\usepackage{bussproofs}
\usepackage{arydshln,stmaryrd,mathpartir}
\usepackage{mdframed}

\newtheorem{myalgorithm}{Algorithm}

\title{Generalization Problems with Atom-Variables in Languages with Binders and Equational Theories} 

\titlerunning{Generalization Problems in Languages with Binders and Equational Theories} 

\author{Daniele Nantes-Sobrinho}{Department of Computing, Imperial College London, UK}{dnantess@imperial.ac.uk}{https://orcid.org/0000-0002-1959-8730}{}

\author{Manfred Schmidt-Schau\ss}{Goethe Universit\"at,  Frankfurt am Main, Germany}{manfredschauss@gmail.com}{https://orcid.org/0000-0001-8809-7385}{}

\author{Alexander Baumgartner}{Universidad de O'Higgins, Rancagua, Chile}{alexander.baumgartner.x@gmail.com}{https://orcid.org/0000-0002-4757-5907}{}
\author{Temur Kutsia}{RISC, Johannes Kepler University, Linz, Austria}{kutsia@risc.jku.at}{https://orcid.org/0000-0003-4084-7380}{}

\authorrunning{D. Nantes-Sobrinho et. al.} %



\begin{document}

\maketitle

\begin{abstract}
  Generalization problems in languages with binders involve computing the most common structure between expressions while respecting bound variable renaming and freshness constraints. These problems often lack a least general solution. However, leveraging nominal techniques, we previously demonstrated that a semantic approach with atom-variables enables the elimination of redundant solutions and allows for computing unique least general generalizations (LGGs).

In this work, we extend this approach to handle associative (A), commutative (C), and associative-commutative (AC) equational theories. We present a sound and weak complete algorithm for solving equational generalization problems, which generates finite weak minimal complete sets of LGGs for each theory.
A key challenge arises from solving equivariance problems while taking into account these equational theories, as identifying redundant generalizations requires recognizing when one expression (with binders) is a renaming of another while possibly considering permutations of sub-expressions. This unexpected interaction between renaming and equational reasoning made this particularly difficult, necessitating semantic tests within the equivariance algorithm.
Given that these equational theories naturally induce exponentially large LGG sets due to subexpression permutations, future work could explore restricted theory fragments where the generalization problem remains unitary. In these fragments, LGGs can be computed efficiently in polynomial time, offering practical benefits for symbolic computation and automated reasoning tasks.


\end{abstract}

\section{Introduction}
\label{sec:intro}

The generalization problem~\cite{Plotkin:1970} (a.k.a. anti-unification problem)  is concerned with finding a generalization of two given input expressions in a formal language, e.g.,  finding a more general program between two different programs, or a more general first-order formula, or a more general proof. Formally, given two expressions $s$ and $t$, the generalization problem is concerned with finding an expression $r$ such that $r\sigma_1=s$ and $r\sigma_2=t$ for some substitutions $\sigma_1$ and $\sigma_2$. For instance, it is easy to see that the expressions $f(a,b)$ and $f(a,g(x,c))$ can be generalized by the variable $z$ (the initial expressions can be retrieved via substitutions $\{z\mapsto f(a,b)\}$ and $\{z\mapsto f(a,g(x,c))\}$). However, this generalizer gives no information about the input expressions. Note that  $f(a, y)$ is also a generalizer, besides, it gives much more information about the initial expressions, in fact, it represents all the similarities. The general approach is to abstract away the particularities of $s$ and $t$ and find a (maximally) common structure between them (in other words, to find a pattern).

As expected, the problem becomes more complicated in languages with binders (e.g., $\lambda$-calculus, the Calculus of Constructions~\cite{DBLP:conf/lics/Pfenning91}, the nominal language~\cite{pitts}, etc), and more sophisticated techniques have been developed. Among the difficulties are the need of: renaming bound variables;  considering various equivalences between expressions such as $\beta$- and $\eta$-conversion; the correct typing, etc. Most of the current frameworks restrict the problem in one way or another, depending on their particular goals.
Extensive research has been done about generalization for higher-order patterns (HOP), which are a restricted class of $\lambda$-terms whose free variables may only be applied to a sequence of distinct bound variables, including equational theories~\cite{DBLP:journals/jar/BaumgartnerKLV17,DBLP:journals/mscs/CernaK20,DBLP:conf/ijcai/CernaK23}. A generic framework for anti-unification of simply typed $\lambda$-terms, which includes HOP generalization as a special case, was studied in \cite{DBLP:conf/rta/CernaK19}.

 This work is about generalization problems for a nominal language with atom-variables and taking into account equational theories for associativity (\A), commutativity (\C) and associativity-commutativity (\AC). The choice is motivated by the expressivity of the language, it can be used to express $\lambda$-terms, or first-order formulas, or $\pi$-calculus processes, etc, and it has notions such as $\alpha$-renaming and freshness embedded in the language, therefore, the usual implicit renaming and freshness choices adopted in many works can be explicitly implemented. In addition, reasoning in the nominal syntax has a first-order-like behaviour that is appealing.

We build upon a framework introduced in~\cite{DBLP:conf/fscd/Schmidt-Schauss22} that uses atom-variables and a semantic approach to discard redundant solutions and compute unique lggs. As expected, working with equational theories requires working in equivalence classes of these theories, which we thought would be a natural extension of the work. Unexpectedly, these equivalence classes interact with our semantic approach when solving the {\em equivariance matching problems} which appear as sub-problems in our algorithm in a non-trivial way that turned out to be a challenging problem. We expand on this and other challenges below.

\noindent {\bf Contributions.}We present a rule-based algorithm, $\enau{E}$, for solving
$E$-generalization problems.
Our approach refines existing structural rules for abstractions and atom-variables while introducing new rules to handle function symbols under equational theories and equivariance. These refinements strengthen the rules, eliminating one of the sources of weak incompleteness.
We propose a method to solve a more general version of the equivariance problem, which consists of finding a permutation $\pi$ and a freshness-context $\nabla$ such that $\nabla\models \pi\cdot s\approx_E t$. Here,  we consider an equational theory $E\in\{\A,\C,\AC\}$, and $\pi,\nabla,s,t$ may contain atom-variables. It requires the analysis to be semantical, represented using the symbol $\vDash$. Note that it is harder to evaluate, e.g., whether the generalized permutation $(A\ C)$ acts in the abstraction such that  $(A\ C)\cdot (\lambda C.\ffac{f} (A, B, C))\approx_{\tt AC}\lambda A.\ffac{f}(A, B, A))$, i.e., makes these two abstractions equal modulo renaming of bound variables and modulo \AC. The difficulty arises because we don't know how $(A\ C)$ acts on $B$. Here, $A, B, C$ are atom-variables which could be instantiated to the same or to different concrete atoms.
%
The valuations of the atom-variables have an effect when branching on the possible permutations of the AC-function symbol $\ffac{f}$. Dealing with the interaction was a challenge, and required specialized methods, described in \Cref{sec:equivariance-new}, to guarantee correctness while branching.

\noindent{\bf Related Work.} Generalization problems have several applications, such as, on parallel recursion scheme detection~\cite{DBLP:journals/fgcs/BarwellBH18}, clone detection \cite{BulychevMinea09}, analogy making~\cite{DBLP:conf/ausai/KrumnackSGK07}, etc.
First-order anti-unification problems modulo equational theories have been extensively investigated~\cite{DBLP:conf/jelia/AlpuenteEEM14, DBLP:journals/iandc/AlpuenteEEM14,DBLP:journals/tcs/Cerna20,DBLP:conf/ijcai/CernaK23}. While the techniques presented there are similar to the ones used here, the interaction with nominal techniques is novel.  Nominal techniques have been investigated in the context of nominal unification~\cite{DBLP:journals/tocl/LevyV12,DBLP:journals/jsc/Schmidt-Schauss19,DBLP:conf/csl/UrbanPG03} including equational theories~\cite{DBLP:journals/mscs/Ayala-RinconSFS21,DBLP:conf/mkm/AyalaRinconFSKN23}, but not much research has been done w.r.t. nominal anti-unification.

In \cite{DBLP:conf/rta/BaumgartnerKLV15}, the authors investigated the problem of computing least general generalizations (lgg) for nominal terms-in-context, i.e., pairs of the form $\tc{\nabla}{s}$, where $\nabla$ is a freshness context and $s$ is a (standard) nominal term.
It was observed that if atoms were taken literally, there does not exist an lgg for two terms-in-context \cite{DBLP:conf/rta/BaumgartnerKLV15}. However, viewing the notation of atoms only as names for atoms and using atom-variables, there exists an lgg and also a minimal set of generalizations \cite{DBLP:conf/fscd/Schmidt-Schauss22,schmidt-schauss-nantes-sobrinho:23}. An alternative approach taken in~\cite{DBLP:conf/rta/BaumgartnerKLV15} is to work with a parametric algorithm that is based on the set of atoms occurring in the problem and computed while searching for the generalizer. We believe that we can extend that algorithm with equational theories, while maintaining the parametric approach in a direct way.

Due to space limitations  proofs and more examples and additional material are given in the Appendix.

\section{Preliminaries}\label{sec:prelim}

We assume the reader is familiar with the nominal syntax, and only recall the main concepts and notations that are needed in the paper. For more details we refer the reader to~\cite{DBLP:journals/iandc/FernandezG07,pitts}.

\subsection{$\grnlang$: A Ground Nominal Language with  Equational Theories.}
Fix a countable infinite set of {\em atoms} $\mathbb{A}=\{a,b,c,\ldots\}$. Atoms are identified by their name, so it will be redundant to say that two atoms $a$ and $b$ are different. A signature $\Sigma$ is a set of {\em term-formers} (disjoint from $\mathbb{A}$) such that to each $f \in \Sigma$
is assigned a unique non-negative integer $n$, called the {\em arity} of $f$, written as $f{:}n$, and {\em constants} ($\const{a,b,c}\ldots$) are symbols in $\Sigma$ with arity zero.
A {\em ground permutation} $\pi$  is a bijection $\mathbb{A} \rightarrow\mathbb{A}$ which is the identity almost everywhere, i.e., the set $\{a \in \mathbb{A} \mid  \pi(a) \neq  a\}$ is finite. Permutations are represented by sequences of {\em swappings} written as $(a \ b )$ where $a,b \in \mathbb{A}$. {$\mathit{Id}$} denotes the identity permutation. The composition of two permutations $\pi$ and $\pi'$ will be denoted as $\pi \circ \pi'$. We may write permutations omitting the $\circ$-symbol as $(a_1~b_1)(a_2~b_2) \ldots (a_n~b_n)$.

The {\em nominal language} $\grnlang$ is built by the grammar:
\[
\begin{array}{l@{\hspace{1cm}}r}  S, S_i ::= a \mid \pi \cdot S \mid f(S_1,\ldots, S_n) \mid \lambda a. S
	& \pi ::= \mathit{Id} \mid (a \ b ) \circ  \pi
\end{array}
\]
where $a,b$ are atoms,  $\abst{a}{S}$ denotes the {\em abstraction} of atom $a$ in the term $S$, and $f(S_1,\ldots, S_n)$ is a {\em function application}, where $f\in \Sigma$ and $f{:}n$.

Since $\grnlang$ consists of ground terms (sometimes called expressions), every term in $\grnlang$ can be simplified by applying permutations such that only terms (and subterms) of the  forms $a \mid f(S_1,\ldots,S_n) \mid \lambda a.S$ remain, i.e. expressions of a lambda-calculus with atoms and function symbols. The effect of a permutation application is defined recursively on the structure of a term as usual.
The predicate $\approx$, which stands for $\alpha$-equivalence of the terms in $\grnlang$, is defined as usual~\cite{DBLP:conf/csl/UrbanPG03,DBLP:journals/tcs/UrbanPG04}:
\[\frac{}{a\approx a}  \qquad \frac{S \approx T}{\lambda a.S\approx \lambda a.T}\qquad \frac{S \approx (a\,b)\cdot T \quad a\# T}{\lambda a.S\approx \lambda b.T} \qquad \frac{S_1\approx T_1, \quad \cdots,\quad S_n\approx T_n}{ f(S_1,\ldots S_n)\approx f(T_1,\ldots,T_n)}\]
where the freshness predicate $\#$ is defined
%
\[
 \frac{}{a\# b} \qquad \frac{}{a\# \lambda a.T} \qquad \frac{  a\# T}{a\# \lambda b.T} \qquad \frac{a \# T_1 \quad \cdots \quad a \# T_n}{a \# f(T_1,\ldots T_n)}
\]
We say $a\#T$ is a {\em freshness constraint} and a set of freshness constraints is a {\em freshness context}. Intuitively, $a\# S$ means that $a$ is {\em fresh} for $S$ (it does not occur free in $S$).
Note that  $a\# a$ is not derivable. 
Finally,  a freshness constraint $a\# S$ {\em holds} (or is {\em valid}) iff it is derivable using the rules.
Validity can be established for $\alpha$-equality constraints in a similar way.

 We consider an extension of the language $\grnlang$ for a signature $\Sigma$ which contains function symbols satisfying an {\em equational theory}, denoted by $E$. In this work we are interested in the theories associativity (\A), commutativity (\!\C) and associativity-commutativity (\AC). The signature is composed of $\Sigma = \Sigma_\emptyset \dotcup \Sigma_{\tt A} \dotcup \Sigma_{\tt C} \dotcup \Sigma_{\tt AC}$, where $\dotcup$ denotes the disjoint union, $\Sigma_\emptyset$ denotes the set of uninterpreted function symbols and $\Sigma_{\tt A} \cup \Sigma_{\tt C} \cup \Sigma_{\tt AC}$ are symbols of arity $2$ that satisfy the respective equational theories. The theories are described by the following (well-known) identities: $ \{f(S_1,S_2) = f(S_2,S_1) \mid f \in \Sigma_{\tt C} \cup \Sigma_{\tt AC}\} \text{\quad (commutativity)}$ and
 $\{f(f(S_1,S_2),S_3) = f(S_1,f(S_2,S_3)) \mid f \in \Sigma_{\tt A} \cup \Sigma_{\tt AC}\} \text{\quad (associativity)}$.

We use the superscript notation $\ffa{f}$, $\ffc{f}$ and $\ffac{f}$ to denote, respectively, $f\in \Sigma_{\tt A}$, $f\in \Sigma_{\tt C}$ and $f\in \Sigma_{\tt AC}$.
We may use a flattened representation of expressions with associative function applications. For example, if $f\in \Sigma_{\tt A}\cup \Sigma_{\tt AC}$, all subterms of the form $f(T_1,\dots,\allowbreak f(S_1,\dots,\allowbreak S_m),\dots,\allowbreak T_n)$ may be written as $f(T_1,\dots,\allowbreak S_1,\dots,\allowbreak S_m,\dots,\allowbreak T_n)$ (if $n+m \ge 2$).

In Figure~\ref{fig:equational_rules} we define the predicate  $\approx_E$ that extends $\alpha$-equivalence of nominal terms by the equality axioms.
This language will be the  semantic base used in the following sections.

\begin{figure}[t!]
\hrule
{\small
\qquad
\begin{minipage}{.45\textwidth}
\begin{prooftree}
\AxiomC{$ T_1 \approx_E S_1 ~\ldots~   T_n \approx_E S_n, n \geq 2$}
\RightLabel{(\A)}
\UnaryInfC{$ \ffa{f}(T_1,\ldots, T_n) \approx_E \ffa{f}(S_1,\ldots, S_n)$}
 \end{prooftree}
 \end{minipage}
\begin{minipage}{.45\textwidth}
 \begin{prooftree}
\AxiomC{$ T_1\approx_E S_2 \quad  T_2\approx_E S_1$}
\RightLabel{(\!\C)}
\UnaryInfC{$ f^C(T_1,T_2)\approx_E f^C(S_1,S_2)$}
\end{prooftree}
 \end{minipage}
\begin{prooftree}
\AxiomC{$ T_1\approx_{E} S_{\rho(1)} \ldots T_n\approx_{E} S_{\rho(n)} ,~ \rho$ is a permutation of $\{1,\ldots,n\}$, $n \geq 2$ }
\RightLabel{(\AC)}
\UnaryInfC{$ \ffac{f}(T_1,~\ldots, ~T_n)\approx_{E} \ffac{f}(S_{\rho(1)},\ldots, S_{\rho(n)})$}
\end{prooftree}
}
\hrule
\caption{Equational Rules. }\label{fig:equational_rules}
\end{figure}

\subsection{\nlat: A Nominal Language with Variables and Equational Theories.}
We recall the basic notions of nominal terms with {\em atom-variables}, which were introduced in~\cite{DBLP:conf/fscd/Schmidt-Schauss22}. We fix a countable infinite set of {\em term-variables} ${\cal X}=\{X, Y, Z,\ldots \}$ and a countable infinite set of {\em atom-variables}  \(\atvar=\{A,B,C, \ldots\}\) such that ${\cal X}\cap \atvar= \emptyset$. Variables are distinct from objects in the
signature $\Sigma$ which is taken from $\grnlang$. The  {\nlat} language is defined by the following grammar of terms (also denoted $T(\Sigma, {\cal X}, \atvar)$) and permutations $\pi$ with atom-variables instead of atoms:
\[
\begin{array}{rl@{\hspace{1cm}}r}
s, s_i &::= W \mid \pi \cdot X \mid f(s_1,\ldots, s_n) \mid  \lambda W. s & \text{(nominal terms)}\\
W, W_i &::= \pi \cdot A & \text{(atom-variable suspensions)}  \\
\pi &::= id \mid (W_1 \ W_2) \circ \pi & \text{(atom-variable permutations)}
\end{array}
\]
Differently from the standard grammar of a nominal language~\cite{DBLP:journals/iandc/FernandezG07}, {\nlat} terms do not contain atoms. We denote the inverse of a permutation $\pi$ with $\pi^{-1}$. If $\pi$ is a list of swappings $(A_1~B_1) \cdots (A_n~B_n)$, then $\pi^{-1} = (A_n~B_n) \cdots (A_1~B_1)$, since $(A_1~B_1) \cdots (A_n~B_n)\cdot (A_n~B_n) \cdots (A_1~B_1)$ is the trivial permutation under any interpretation of atom-variables. There are two kinds of {\em suspensions}: $\pi\cdot A$ and $\pi\cdot X$, i.e., one on atom-variables and one on term-variables, respectively. We write $\pi\cdot V$ to denote a suspension over a variable $V\in \atvar \cup {\cal X}$, and $\AV(F)$ for the set of atom-variables occurring in a the formula $F$. {\em Abstractions} are more general here as well:  $\lambda W.s$ denotes the abstraction of a (suspended) atom-variable $W$ on a term $s$. Function application extends to nominal terms in {\nlat} as expected. As in $\grnlang$, in~\nlat~we permit function symbols that are commutative, associative, or both.

\begin{example} Simple examples of \nlat-terms are: the abstraction of an atom variable on a term $\lambda A. f(A)$; a permutation of atom-variables $(A \ B)$ that is suspended on an atom-variable $C$, that is, $(A \ B)\cdot C$; and finally, a permutation of atom-variables suspended on a term variable  $(A \ B)\cdot X$. Note that the semantics of these terms can only be established once the atom-variables $A,B,C$ are instantiated to concrete atoms and the term-variable is instantiated to some term. More complex examples are:  $ \lambda (A \ B) \cdot C. f(A, (B \ C) \cdot X)$ and $ \lambda A.f( (A \ B) (B \ C)\cdot X,  B)$.
\end{example}
Substitutions, ranged over by $\sigma, \gamma, \ldots$, are  maps from term-variables to terms in {\nlat} and from atom-variables to suspensions of atom-variables, such that the {\em domain},  $\mdom(\sigma)=\{V\in \atvar\cup {\cal X} \mid V\sigma \neq V\}$ is finite.
The {\em action of a substitution} on a term $t$, denoted as $t\sigma$, is inductively defined as usual.
 Substitutions are written in postfix notation: $t(\sigma\gamma) = (t\sigma)\gamma$.

We will also use \nlat-freshness constraints of the form $W\#s$ where $s$ is an \nlat-term and \nlat-freshness contexts ($\Delta, \nabla, \ldots$), which are conjunctions (written as sets) of \nlat-freshness constraints. Note that freshness constraints of the form $A\#s$ are
sufficient, since $(W_1\ W_2){\cdot}A\#s$ is equivalent to $A\#(W_1\ W_2) {\cdot}s$. Thus, constraints with permutations such as $A\#\lambda B.((A C)(D~B)\cdot A)$
are possible.

 An {\em \nlat-term-in-context} $(\nabla,s)$ is a pair of an \nlat-freshness context $\nabla$ and an \nlat-term $s$.

\section{The Semantics of \nlat-Expressions and Freshness Constraints}\label{ssec:semantics}

The semantics of \nlat-expressions is based on interpretations $\rho$ to ground expressions in $\grnlang$.

\begin{definition}[Interpretations]\label{def:sem_fun} 
    An interpretation $\rho$ maps \nlat-terms to $\grnlang$-terms.
    It is determined by a mapping from atom-variables to concrete atoms and term-variables to (ground) $\grnlang$-terms.
Interpretations act on \nlat-terms homomorphically.
It is also extended to freshness constraints, contexts, and equational theories, as expected.
\end{definition}

We define now the semantics, which includes  equivalence of expressions, validity of freshness constraints and contexts. We also assume that  an equational theory $E$ is defined by axioms from \A, \AC, and \!\C~for specific function symbols. Like substitution applications,  applications of  interpretations are written in postfix notation.  

\begin{definition}[Semantics in \nlat]\label{def:sem_rel} The semantics of~\nlat~is defined on expressions in $\grnlang$: \hfill
\begin{enumerate}
    \item\label{def:interp-1} Two \nlat-permutations  $\pi_1,\pi_2$ are {\em semantically equivalent} iff for all interpretations $\rho$: $\pi_1\rho$ is the same permutation as  $\pi_2\rho$.  This is denoted as $\pi_1 \equiv \pi_2$

    \item\label{def:interp-2} Two \nlat-terms  $s,t$ are {\em semantically equivalent}
    iff for all interpretations $\rho$: $s\rho \approx_E t\rho$.
       \item\label{def:interp-3} An \nlat-freshness constraint $W\#s$ is {\em valid}, iff for all interpretations $\rho$: $W\rho\#s\rho$ is valid.
    \item\label{def:interp-4} A  context $\nabla$ is {\em valid}, if for all interpretations $\rho$, every constraint in $\nabla\rho$ holds.
     \item\label{def:interp-5} A  context $\nabla$ is {\em consistent} w.r.t. $\rho$  iff  every constraint in $\nabla\rho$ holds. In this case, we say that $\rho$ {\em respects} $\nabla$.
\end{enumerate}
\end{definition}

Note that items \ref{def:interp-1},\ref{def:interp-3},\ref{def:interp-4},\ref{def:interp-5} are independent of the equational theory.

\begin{example}\label{ex:abstraction_constraint}
Interesting constraints such as $A\#\lambda B.A$ or $ C\# \lambda A. \lambda B. C$ are expressible in \nlat.    Note that $A\#\lambda B.A$ is consistent w.r.t. $\rho$ iff  $A\rho = B\rho$. Differently,  $C\# \lambda A. \lambda B. C$ is consistent w.r.t. $\rho$ iff $ C\rho = A\rho \vee C\rho = B\rho$: in fact, for an interpretation $\rho$ of $A,B,C$ such that $C\rho \not= B\rho \wedge C\rho \not= A\rho$, the constraint would reduce to  $C\rho\#C\rho$, which is inconsistent.
\end{example}

We can now consider {\em judgements} of the form $\nabla \vDash A\# s, \nabla\vDash s\approx_E t$ and $\nabla \vDash \pi_1\equiv\pi_2$, and their meaning is established next.

\begin{definition}[Semantics of Judgements]\label{def:semantics_judgements}
    Let $\nabla$ be a consistent \nlat-freshness context, $A\#s$ a freshness constraint and $\pi_1,\pi_2$ be \nlat-permutations. We say that
    \begin{itemize}
    \item $\nabla\vDash A\#s$ {\em holds} iff for all interpretations $\rho$: $\nabla\rho$ valid $\implies$  $A\rho\# s\rho$  valid.
    \item $\nabla\vDash s\approx_E t$ {\em holds} iff for all interpretations $\rho$: $\nabla\rho$ valid $\implies$  $s\rho\approx_E t\rho$  valid.
    \item $\nabla\vDash \pi_1 \equiv  \pi_2$ {\em holds} iff for all interpretations $\rho$: $\nabla\rho$ valid $\implies$  $\pi_1\rho = \pi_2\rho$.
    \end{itemize}
\end{definition}


\section{Generalization Problems in~\nlat~modulo Equational Theories}


We are interested in the following problem:
\begin{mdframed}
\begin{description}
\item[{\bf Given:}]  Two \nlat-terms-in-context $(\nabla,s)$ and $(\nabla,t)$.
\item[{\bf Find:}] An \nlat-term-in-context $\tc{\Gamma}{r}$ that is a generalization modulo $E$ of $\tc{\nabla}{s}$ and $\tc{\nabla}{t}$ which preserves most of their common structure while taking into account the equational theory $E$.
\end{description}
\end{mdframed}
This problem is known as the {\em generalization problem modulo $E$} (also known as the {\em anti-unification problem modulo $E$}), which we will investigate in the context of the language~\nlat. In the following sections, we will introduce the main ingredients to solve this problem, which we will refer to simply as generalization (or anti-unification), as reasoning modulo $E$ will be used throughout this work.

\subsection{A Semantic Approach to Generalization in~\nlat}
We follow a semantic approach to solve the generalization problem in~\nlat~by investigating the semantics of the pairs  $(\nabla,s)$ and $(\nabla,t)$ that we want to generalize. The semantics will be based on the equivalence classes  modulo $\approx_E$ (i.e. modulo $E$ and $\alpha$-renaming) of the ground instances $s$ and $t$ that respect the corresponding instances of $\nabla$. First, we define the semantics of a pair $(\nabla,s)$ w.r.t. interpretations of $s$ on the ground language $\grnlang$:

\begin{definition}\label{Def:semantics-term-in-context}
    The {\em semantics of} $(\nabla,s)$, denoted $\sem{\nabla}{s},$ consists of the following set:
    $$\sem{\nabla}{s}=\{[s\rho]_{E} \mid   \rho \text{ is an interpretation and }  \nabla \rho \text{ holds}\},$$
    where $\eqclass{r}$ denotes the equivalence class of $r$ modulo $\approx_E$.
\end{definition}

Note the semantics of terms-in-context that contain abstractions are always infinite due to the existence of infinitely many names for $\alpha$-renaming:
\begin{example}[Semantics of Abstractions - part 1]
    The semantics of  $(\{A\#B\}, \lambda A. f(A,B))$ is the following:
    $$
    \begin{aligned}
    \sem{\{A\#B\}}{\lambda A. f(A,B)}&=\left\{ \eqclass{\lambda A\rho. f(A\rho,B\rho)} \mid \rho \text{ is an interpretation}\right.
     \hspace{0cm }\left. \text{ and } A\rho \# B\rho \text{ holds}\right\}\\
    &=\{\eqclass{\lambda A\rho. f(A\rho, B\rho)} \mid A\rho \neq B\rho\}
    = \{\eqclass{\lambda a.f(a,b)},\eqclass{\lambda a.f(a,c)},\ldots\}
    \end{aligned}
    $$
 The second argument of $f$ can be any concrete atom, as long as it is not captured by the abstraction.
\end{example}

\begin{example}[Semantics of Abstractions - part 2]
    The semantics of the term-in-context $(\{A\#C\}, \lambda A. f(A,B))$ is the following:
    $$
    \begin{aligned}
    \sem{\{A\#C\}}{\lambda A. f(A,B)}&=\left\{ \eqclass{\lambda A\rho. f(A\rho,B\rho)} \mid  \rho \text{ is an interpretation}\right.
     \hspace{0cm }\left. \text{ and } A\rho \# C\rho \text{ holds}\right\}\\
    &=\{\eqclass{\lambda A\rho. f(A\rho, B\rho)} \mid A\rho \neq C\rho\}\\
    \end{aligned}
    $$
 the {\semfunappname} $C\rho$ is irrelevant for the semantics since it does not occur in the term. One can always find $\rho$ such that $ A\rho \neq C\rho$.
 Hence, $(\{A\#C\},\lambda A. f(A,B))$ is equivalent to $(\emptyset,\lambda A. f(A,B))$
 \end{example}

The semantics of terms-in-context provides a framework to  for comparing terms and determining their relative generality. Specifically, a pair $(\nabla, t)$ is {\em more general w.r.t. $E$ than}  $(\Delta, s)$, denoted $(\nabla, t)\epreceq (\Delta, s)$,   iff $\sem{\Delta}{s}\subseteq \sem{\nabla}{t}$. They are {\em equivalent} iff both $(\nabla, t)\epreceq (\Delta, s)$ and $(\nabla, t)\succeq_E (\Delta, s)$ hold, meaning that their semantics are identical $\sem{\Delta}{s}= \sem{\nabla}{t}$. This allows us to formally define an
generalization and {\em the least general generalization} ($lgg_E$) of two terms-in-context.

\begin{definition}[Generalization, $lgg_E$]  A term-in-context $(\nabla,r)$ is an {\em generalization of}  $(\nabla_1,s)$ and $(\nabla_2,t)$ iff $(\nabla,r)$ the following holds: $(\nabla, r)\epreceq (\nabla_1, s)$ and $(\nabla, r)\epreceq (\nabla_2, t)$.
 The pair $(\nabla,r)$ is called a {\em least general generalization} (\lgge) of $(\nabla_1,s)$ and $(\nabla_2,t)$ iff it is a `smallest' generalization, i.e., every other generalization $(\nabla',r')$ of $(\nabla_1,s)$ and $(\nabla_2,t)$, is more general than $(\nabla,r)$, that is,   $(\nabla',r')\epreceq(\nabla,r)$ holds.
\end{definition}

A set $G$ of generalizations for  $(\nabla_1,s)$ and $(\nabla_2,t)$ is {\em complete} iff for each generalization $(\Delta,r')$ of  $(\nabla_1,s)$ and $(\nabla_2,t)$, there exists
$(\nabla,r)\in G$ such that $(\Delta,r')\epreceq (\nabla,r)$.
The set $G$ is called {\em minimal}, if for all different $(\nabla', r_1), (\nabla'',r_2)\in G$, we have  $(\nabla', r_1)\not\!\epreceq (\nabla'',r_2)$ and $(\nabla'', r_2)\not\!\epreceq (\nabla',r_1)$.
Such minimal set will be denoted as $\mcsg_E((\nabla_1,s),\allowbreak (\nabla_2,t))$.

The existence and cardinality of the minimal complete set of generalizations characterize the anti-unification type of an equational theory $E$. Thus, $E$ has \emph{nullary} (anti-unification) type if there are terms-in-context $T_1$ and $T_2$ for which $\mcsg_E(T_1,T_2)$ does not exist. Otherwise, $E$ has type
 \emph{unitary} (resp. \emph{finitary}), if $\mcsg_E(T_1,T_2)$ is a singleton (resp. finite) for all terms-in-context $T_1$ and $T_2$;
 and \emph{infinitary}, if $\mcsg_E(T_1,T_2)$ is infinite for some $T_1$ and $T_2$.

\section{ $\enau{E}$: An Anti-Unification Algorithm}
\label{sec:enau}

In this section, we introduce the rule-based algorithm $\enau{E}$ to solve the generalization problem. The rules for $\enau{E}$ are described in  two parts: (i) Figure~\ref{fig:nau_rules} contains the basic rules that have no explicit interaction with $E$; (ii) Figure~\ref{fig:enau_rules} contains the rules dealing with $E$. In addition, $\enau{E}$ relies on a {\em novel} algorithm \eqvm~for solving the {\em equivariance problem modulo $E$}, which is detailed in \Cref{sec:equivariance-new}. Rules in~Figure~\ref{fig:nau_rules}  labeled in \aoldrule{gray} remain unchanged from~\cite{DBLP:conf/fscd/Schmidt-Schauss22}, as they required no modification to accommodate the equational theories in $E$. In contrast, rules labeled in $\anewrule{green}$ are either new or improved versions of existing rules.
The symbol $\dotcup$ denotes disjoint union.



The $\enau{E}$ operates on tuples (called states) of the form $P;\; S ;\; \Gamma;\;  \sigma$
where:
\begin{itemize}
\item $P$ denotes the set of unsolved {\em anti-unification equations} between terms, denoted $X\colon t\triangleq s$. Here, $X$ is a term-variable that will represent the generalization of $t$ and $s$.
\item $S$ denotes the set of solved anti-unification equations (the store). The store $S$ may contain anti-unification equations between suspended atom-variable permutations, denoted  $C\colon \pi\cdot A_1\triangleq \pi_2\cdot A_2$, where $C$ is an atom-variable that will represent the generalization of $\pi\cdot A_1$ and $\pi_2\cdot A_2$.
\item $\Gamma$ denotes the freshness context computed so far which constrains generalization variables;
\item $\sigma$ is a substitution, computed so far, mapping generalization variables to terms in \nlat.
\end{itemize}
Furthermore, there is as a global parameter a freshness context $\nabla$ that does not constrain generalization variables.
Note that in $P$, generalization variables are only term-variables, while in $S$ they can be term- or atom-variables.
$X$ and $C$ are called {\em generalization variables}.

\begin{figure}[!t]\label{fig:enau_rules}
\small{
\[
\begin{array}{|l|}
\hline
\aoldrule{Dec} \ \text{\bf  Decomposition}\\
 \{X\!:  f\ttup{t}{m}  \triangleq  f\ttup{s}{m} \}\dotcup P; S; \Gamma; \sigma \Lra     \bigcup_{i=1}^m\{Y_i\!: \! t_i\triangleq s_i\}\cup P; S;  \Gamma ;
      \sigma \{X \mapsto f\ttup{Y}{m}\}\\[1mm]
    \text{where }
    f\in \Sigma \text{ and }
    \ttup{Y}{m}=Y_1,\ldots,Y_m \text{ are fresh variables}, m\ge 0.\\[2mm]
\aoldrule{Abs} \text{\bf Abstraction}\\
 \{ X :  \abst{W_1}{t} \triangleq \abst{W_2}{s} \}\dotcup P;S ; \Gamma;  \sigma \Lra   \\
\hspace{0.5cm}     \{ Y: \swap{W_1}{C}\permef t \triangleq \swap{W_2}{C}\permef s \}\cup P;  S;    \Gamma\cup\{C\#\lambda W_1.t, C\#\lambda W_2.s\}; \sigma\{X\mapsto \lambda C.Y\},\\[1mm]
\text{where }
 Y \text{ is fresh and  }
 C \text{ is a new atom-variable}. \\[2mm]
 \aoldrule{SusAA} \ \text{\bf Suspension on Atom-Variables}\\
 \{ X :  \pi_1{\cdot}A \triangleq \pi_2{\cdot}B \}\dotcup P;S ; \Gamma;  \sigma \Lra
     P;  S;    \Gamma; \sigma\{X\mapsto \pi_1{\cdot}A\},\quad
\text{if } \
\Gamma \models  \pi_1{\cdot}A =  \pi_2{\cdot}B.
\\[2mm]
 \aoldrule{SusYY} \ \text{\bf Suspensions on $Y$}\\
 \{ X :  \pi_1{\cdot}Y \triangleq \pi_2{\cdot}Y \}\dotcup P;S ; \Gamma;  \sigma \Lra
     P; S;    \Gamma; \sigma\{X\mapsto \pi_1{\cdot}Y\},\quad
\text{if } \
\Gamma \models  \pi_1{\cdot Y}  =  \pi_2{\cdot Y}.
\\[3mm]
 \anewrule{ SolAB} \ \text{\bf Solve different atom-suspensions}\\
 \{X :  \pi_1{\cdot}A \triangleq \pi_2{\cdot}B \}\dotcup P;S ; \Gamma;  \sigma \Lra   \\
\hspace{0.5cm}    P;  S \cup \{C: \pi_1{\cdot}A \triangleq \pi_2{\cdot}B  \};    \Gamma \cup \{C\#\lambda \pi_1{\cdot}A.\lambda \pi_2{\cdot}B.C \}; \sigma\{X\mapsto C\},\\[1mm]
\text{if } \
\Gamma \not \models  \pi_1{\cdot}A =  \pi_2{\cdot}B \text{ and } C \text{ is a new atom-variable}.
\\[2mm]
\aoldrule{Sol} \ \text{\bf Solve}\\
 \{{ X:  t \triangleq s}\}\dotcup P; S; \Gamma;\; \sigma \Lra
     P; S\cup\{{ X:  t \triangleq s}\};\Gamma ;\sigma\quad\\[1mm]
     \text{where none of the previous rules is applicable.}
     \\[2mm]
\anewrule{Mer}\ \text{\bf  Merge}\\
 \emptyset; \{{ Z_1 : t_1 \triangleq s_1, Z_2 : t_2
    \triangleq s_2}\}\dotcup S ; \Gamma; \sigma \Lra
     \emptyset;  \{{ Z_1: t_1 \triangleq s_1}\}\cup S ;
     \Gamma{\{Z_2 \mapsto \pi\permap Z_1\}}; \sigma\{ Z_2 \mapsto \pi\permef  Z_1\},\\[1mm]
\text{where }  \pi=\eqvm(\{t_1 \match t_2, s_1 \match s_2\}, \Gamma) \text{ and } (Z_1,Z_2) \text{ is either a pair of atom variables or}\\
\text{a pair of term variables.}\\
\hline
\end{array}
\]
}
\caption{$\enau{}$ transformation rules for $E=\emptyset$.
}\label{fig:nau_rules}
\end{figure}
\normalsize

\subparagraph*{Idea of the algorithm.} Given \nlat-terms-in-context $\tc{\nabla}{t}$ and $\tc{\nabla}{s}$, the (non-deterministic) algorithm for $\enau{E}$ starts with $(\{X\colon t\triangleq s\};\; \emptyset;\; \nabla;\; Id)$ and applies the rules in Figure~\ref{fig:nau_rules} and Figure~\ref{fig:enau_rules} exhaustively, until we reach a final state $(\emptyset;S;\Gamma;\sigma)$. The sequence of rule applications is called a ($\enau{}$-){\em derivation}.
The generalization variable $X$ occurs neither in $\nabla$, nor in $t$, nor in $s$.
Intuitively, it represents the generalization that becomes less and less general as the algorithm advances (by applying some rules).
The result of this procedure is of the form $\emptyset;\; S;\; \Gamma ;\; \sigma$,
the \emph{computed generalization is} $\tc{\Gamma}{X\sigma}$ and $S$ contains the differences of $t$~and~$s$.
To compute the set of all results, the contribution of all non-deterministic computation paths has to be gathered in a set.


 \subparagraph*{Description of the \aoldrule{unmodified} rules in Figure~\ref{fig:nau_rules}.}  Rule \aoldrule{Dec} is standard and applies when  both terms have a common head function symbol. It decomposes the problem into several subproblems and extends the substitution with a mapping that specifies the generalization. Rule \aoldrule{Abs}
 generalizes abstractions on $W_1,W_2$ by choosing a new atom-variable $C$ that was not used anywhere else, and swapping both $W_1, W_2$ with $C$ on the subterms underneath the abstraction. Note that the freshness context is extended with constraints that might not allow simplifications (relying on semantics). The substitution is extended with the current generalization variable, say $X$, specialized to an abstraction on $C$. Rule \aoldrule{SusAA} is about generalization of semantically-equal suspensions of atom-variables. The rule
  removes the problem and extends the substitution with a mapping to this suspension. Rule \aoldrule{SusYY} is similar to \aoldrule{SusAA}, but applies to suspensions on the same term-variable.  Rule \aoldrule{Sol} applies when none of the previous rules can be applied. It only moves the problem to the store and does not alter the substitution.

  \begin{example}[For abstraction]
 Consider $P=\{X: f(\lambda A.A,Z)\triangleq f(\lambda B.C,Z)\}$   and $\nabla=\emptyset$. The $\enau{\emptyset}$-derivation starts with  $ \{X: f(\lambda A.A,Z)\triangleq f(\lambda B.C,Z)\};\emptyset;\emptyset;Id$, after $ {\arule{Dec}}$, $
     \{X_1: \lambda A.A\triangleq \lambda B.C , X_2: Z\triangleq Z\};\emptyset;\emptyset;\{X\mapsto f(X_1,X_2)\}$, after an application of $\arule{Abs}$, we get
      $\{Y:   D\triangleq (D\ B)\cdot C , X_2: Z\triangleq Z\};\emptyset;\{D\# \lambda A.A,D\# \lambda B.C\};\{X\mapsto f(\lambda D.Y,X_2)\}$.
 \end{example}

  \subparagraph*{Description of the \anewrule{new/improved} rules in Figure~\ref{fig:nau_rules}.} Rule \anewrule{SolAB} applies when the context cannot prove that the atom-variable suspensions are the same: we choose a new atom-variable $C$ and move the anti-unification equation to the store with $C$ as its generalization variable. Furthermore, we introduce the constraint $C\# \lambda \pi_1\cdot A.\lambda \pi_2\cdot B.C$ which is equivalent to $C=\pi_1\cdot A$ or $C=\pi_2\cdot B$. This constraint collects information from the store that $C$ can only be one of the two suspensions, providing a more specific generalization. This restriction is stronger than the initial version introduced~\cite{DBLP:conf/fscd/Schmidt-Schauss22}.However, as we will see in Example~\ref{ex:weak}, further strengthening is required to guarantee the completeness of  $\enau{E}$.

Rule \anewrule{Mer} applies when the set of problems is empty, and it identifies  if there are no redundancies in the store, i.e., the same solved problem (possibly with some renaming) under different generalization variables. Thus, given two problems in the store $\{{ Z_1 : t_1 \triangleq s_1, Z_2 : t_2 \triangleq s_2}\}$ the goal of \eqvm~is to find a permutation $\pi$ such that $\pi\cdot t_1\approx_E s_1$ and $\pi\cdot s_1\approx_E s_2$. If such permutation exists, both problems are merged into one, and the appropriate updates are done on the context and substitution, e.g. the substitution must be updated with the information $Z_2\mapsto \pi\cdot Z_1$.
The rule \anewrule{Mer} relies on a new version (defined in~\Cref{sec:equivariance-new}) of the algorithm
(\eqvm) for  {\em equivariant matching  modulo $E$} for \nlat$_E$, which respects the current freshness constraints, to identify if there is a permutation of atom-variables that maps one solved term to another, while taking into account the equational theories in $E$. A new algorithm is necessary due to the unexpected interaction between suspensions of atom-variables and equational theories when computing the permutation $\pi$, we give more details in \Cref{sec:equivariance-new}, and illustrate it in Example~\ref{ex:equivariance_e}.


 \subparagraph*{Description of the \anewrule{new} rules in Figure~\ref{fig:enau_rules}.} In~\Cref{fig:enau_rules},  $Y_1, Y_2$ are fresh variables of the corresponding sorts,  $1\leq k< n$, $1\leq l< m$ and $ n,m\geq 2$. We assume that \A~and \AC~function applications are flattened. Also, both $\ffac{f}(t)$ and $\ffa{f}(t)$ simply stand for $t$. Write $\ffa{f}(\overline{t})_k$ to denote $\ffa{f}(t_1,\ldots, t_k)$, and $\ffa{f}(\overline{t})_{k+1\ldots n}$ denotes $\ffa{f}(t_{k+1},\ldots, t_n)$. Similarly, for $\ffac{f}(\overline{t})_k$ and $\ffac{f}(\overline{t})_{k+1\ldots n}$.
The rules are as expected and illustrated with examples:

 \begin{figure}[!t]
{\small
\[
\begin{array}{|l|}
\hline
    \anewrule{DecA} \  \text{\bf Associative Decomposition}\\
    \{X:  \ffa{f}\ttup{t}{n} \triangleq \ffa{f}\ttup{s}{m}\}\dotcup P; S; \Gamma; \sigma \Lra 
    \left\{\begin{aligned}
    &Y_1:  \ffa{f}\ttup{t}{k}\triangleq \ffa{f}\ttup{s}{l},\\
    & Y_2 : \ffa{f}\ttup{t}{{k+1}..n} \triangleq \ffa{f}\ttup{s}{{l+1} ..m}\end{aligned}\right\}\cup P; S; \Gamma ; \sigma \{X \mapsto \ffa{f}(Y_1,Y_2)\}\\[1mm]
        \text{if } 1 \leq k < n; 1 \leq l < m
        \\[2mm]
      \anewrule{DecC}  \  \text{\bf Commutative Decomposition}\\
\{X:  \ffc{f}(t_1,t_2) \triangleq \ffc{f}(s_1,s_2)\}\dotcup P; S; \Gamma; \sigma \Lra   
\{Y_1:  t_1\triangleq s_i, Y_2 : t_2 \triangleq s_{3-i}\}\cup P; S;  \Gamma ;
      \sigma \{X\mapsto \ffc{f}(Y_1,Y_2)\}\}\\[2mm]
      \anewrule{DecAC} \ \text{\bf Associative-Commutative Decomposition}\\[1mm]
  \{X:  \ffac{f}\ttup{t}{n} \triangleq \ffac{f}\ttup{s}{m}\}\dotcup P; S; \Gamma; \sigma \Lra 
  \left\{
  \begin{aligned}
 & Y_1:  \ffac{f}\ttup{t_{i}}{k}\triangleq \ffac{f}\ttup{s_{j}}{l}, \\
  & Y_2 : \ffac{f}\ttup{t_{i}}{{k+1}..n} \triangleq \ffac{f}\ttup{s_{j}}{l+1..m}
  \end{aligned}\right\}\cup P; S;  \Gamma ;\sigma \{X\mapsto \ffac{f}(Y_1,Y_2)\}\\[1mm]
\text{where } \{i_1\dots,i_n\}\approx\{1,\dots,n\},\{j_1\dots,j_m\}\approx\{1,\dots,m\} \\[1mm]
\hline
\end{array}
\]
}
\caption{Equational Rules for $\enau{E}$}\label{fig:enau_rules}
\end{figure}


%

\begin{example} 
\label{ex:first_ac}
    Consider the problem $P=\{ X: \lambda A. \ffac{f}(A,A,B)\triangleq \lambda B. \ffac{f}(A,B,A)\}$ and with $\nabla=\{A\#B\}$ and the  $\enau{E}$ derivation starting from  $P; \emptyset; \nabla;Id$. Applying $\arule{Abs}$ we get:
    $$ \{Y:  \ffac{f}(C,C,(A\ C)\cdot B)\triangleq \ffac{f}((B\ C)\cdot A,C,(B\ C)\cdot A)\};\emptyset; \Gamma;\{X\mapsto \lambda C. Y\}$$
     where  $\Gamma =\{A\# B, \allowbreak C\#\lambda A. \ffac{f}(A,A,B), C\#\lambda B. \ffac{f}(A,B,A) \}$.
     After the application of $\arule{DecAC},2\times\arule{DecAC},3\times \arule{SolAB}$ one the branches contains\footnote{Full computation in the Appendix, Example~\ref{ex:first_ac2}} the state:
    \[\small
     \emptyset; \left\{
        D:  C \triangleq (B\ C)\cdot A,
        D':  C \triangleq (B\ C)\cdot A,
        D'': (A\ C)\cdot B\triangleq C
    \right\}; \Gamma';\{X\mapsto \lambda C. \ffac{f}(D,D',D'')\}
    \]
     where $\Gamma'$ (omitted here) extends $\Gamma$ with the constraints for $D$, introduced by \anewrule{SolAB}.

    The application of $\arule{Mer}$ requires computation  $\eqvm(\{C\match (A\ C)\cdot B, (B\ C)\cdot A\match C\}, \Gamma' )$. The result is:
    $  \emptyset; \left\{
        D:  C \triangleq (B\ C)\cdot A
    \right\}; \Gamma'';\{X\mapsto \lambda C. \ffac{f}(D,D,\pi \cdot D)\}$,
    where  $\Gamma'' =\Gamma \cup \{D\#\lambda C.\lambda (B\ C)\cdot A.D, \; \allowbreak \pi\cdot D\# \lambda (A\ C)\cdot B.\lambda C.\pi\cdot D \}$.

    In \Cref{ex:first_equiv} we compute  $\eqvm(\{C\match (A\ C)\cdot B, (B\ C)\cdot A\match C\}, \Gamma' )=\swap{A}{C}\circ \swap{C}{B}$  which we denote as $\pi$ above. Note that $\pi\cdot D$ is a suspension, we only apply the permutation after instantiating $D$ according to the store.  The generalization obtained is the term-in-context $\tc{\Gamma'}{\lambda C. \ffac{f}(D,D,\pi \cdot D)}$. Note that this is one possible generalization in one branch of the generalization tree. A different branch may produce a different generalization.
\end{example}

\begin{example}[Different $\lggeq{A}$'s]
Consider the problem $X:\ffa{f}(\const{a},\const{b},g(Y))\triangleq \ffa{f}(\const{c},g(Z))$, with $\const{a,b}$ and $\const{c}$ constants.
By using rule $\arule{DecA}$  we  reach the state $\{X_1:\const{a}\triangleq \const{c}, X_2:\ffa{f}(\const{b},g(Y))\triangleq  g(Z)\};\emptyset;\emptyset;\{X\mapsto f(X_1,X_2)\}$ in one branch, and no further decomposition rule can be applied. This branch would give the generalization $f(X_1,X_2)$.
However, by grouping the arguments of $\ffa{f}$ in a different  way, we get the  state
$ \{X_1:\ffa{f}(\const{a,b})\triangleq \const{c}, X_2:g(Y)\triangleq  g(Z)\};\emptyset;\emptyset;\{X\mapsto f(X_1,X_2)\}$
in another branch. We can apply  $\arule{Dec}$ rule in the problem labelled by $X_2$, and this would give us a more specific generalizer $f(X_1, g(Y_1))$.

 The problem  $X:\ffa{f}(\const{a},h(Z), \const{b,b'}, g(X))\triangleq \ffa{f}(\const{c,d},h(Z_1),g(Z_2))$ is also interesting: in one branch the algorithm will compute an \A-generalizer with an occurrence of function symbol
 $h$, an expression like $\ffa{f}(X', h(X''), X''')$; in another branch an \A-generalizer with an occurrence of function symbol $g$, say  $\ffa{f}(X', Y', g(X''))$.

\end{example}

Associative decomposition will lead to some redundancies and the computed set of generalizations should be further minimized. This behaviour is similar to the first-order \A-anti-unification~\cite{DBLP:conf/jelia/AlpuenteEEM14} and HOP \A-anti-unification~\cite{DBLP:journals/mscs/CernaK20}.
\subsection{Correctness Results}\label{ssec:correctnesss}

All the results here are parametric on an  equational theory  $E$ that is either $\emptyset$ or \A, \C~or \AC. In the case $E=\emptyset$, all the results were verified in~\cite{DBLP:conf/fscd/Schmidt-Schauss22}, we will recall the proof ideas here, but for self-containedness, but we will give more details for the cases in which $E\neq \emptyset$.
\begin{theorem}[Termination]
$\enau{E}$ is terminating.
\end{theorem}
\begin{proof}
    The proof is standard and follows by analyzing the rules in \Cref{fig:nau_rules} and \Cref{fig:enau_rules}, accompanied by a lexicographic ordering. The equational rules induce a lot of branching, but after each step the initial problem is replaced by `smaller' problems.
\end{proof}
\begin{theorem}[Soundness]\label{theo:soundness}
Given terms-in-context $\tc{\nabla}{s}$ and $\tc{\nabla}{t}$. If there exists a derivation
$ \{X:t \triangleq s\};\emptyset;\emptyset;Id \naustep{+}\emptyset;S;\Gamma;\sigma$
 obtained by an execution of $\enau{E}$, then
$\tc{\Gamma}{X\sigma}$ is an $E$-generalization of
$\tc{\nabla}{ t}$ and $\tc{\nabla}{s}$.
\end{theorem}
\begin{proof}
Let $P,S,\Gamma;\sigma \Longrightarrow P',S',\Gamma';\sigma'$ be a rule application of $\enau{E}$.
    By inspecting each rule it can easily be verified that the following invariants hold:
(i) The generalization variables of $P'$ and $S'$ are pairwise distinct;
(ii) After each rule application $P'\cup S'$ contains all the information, such that, two substitutions $\sigma_1$ and $\sigma_2$ can be obtained to ``reverse'' the rule application. I.e., by instantiating a generalization variable with $\sigma'$ followed by applying $\sigma_1$ or $\sigma_2$, respectively, we can obtain the original terms from $P\cup S$.

By transitivity follows that $(\emptyset, X\sigma)$ is an $E$-generalization.
\end{proof}

 \begin{example}[Weak Completeness]\label{ex:weak}
 For $P=\{ X: f(c_1,A) \triangleq f(c_2,A)\}$ and $\nabla=\emptyset$, the $\enau{\emptyset}$-algorithm terminates with configuration $\emptyset; \{Y_1: c_1\triangleq c_2\};\emptyset;\{X\mapsto f(Y_1,A)\}$. Note that $[f(a,a)]$ is in $ \sem{\emptyset}{f(Y_1, A}$, but it is neither in $\sem{\emptyset}{f(c_1,A)}$ nor in $\sem{\emptyset}{f(c_2,A)}$. It is a generalizer, but not the lgg: the lgg is $(\{A\#Y_1\},\allowbreak f(Y_1,A))$, and the added constraint forces that $Y_1$ and $A$ differ for every valid interpretation (cf. \Cref{ex:abstraction_constraint}).
 \end{example}

\begin{theorem}[Weak Completeness]
Given pairs $\tc{\nabla}{s}$ and $\tc{\nabla}{t}$. If $\tc{\Delta}{r}$ is an generalization of $\tc{\nabla}{t}$ and $\tc{\nabla}{s}$, then there exists a derivation
$\{X:t \triangleq s\};\emptyset;\emptyset;Id \naustep{+}\emptyset;S;\Gamma;\sigma$
obtained by an execution of $\enau{E}$ and a freshness context $\Gamma'$, such that $\tc{\Delta}{r}\epreceq \tc{\Gamma \cup \Gamma'}{X\sigma}$.
\end{theorem}

\begin{proof}
    The proof of completeness is by induction on the structure $r$, and follows by analyzing the rule of $\enau{E}$ applied in the problem.  The proof follows the same lines of the proof in~\cite{DBLP:conf/fscd/Schmidt-Schauss22}. Additional cases have to be considered when the terms in the problem are headed with function symbols $f^E$ with equational theories. Since the rules in Figure~\ref{fig:enau_rules} expand all possible branches, completeness up to freshness context is guaranteed.
\end{proof}


It is a consequence of the weak completeness that the process of computing the freshness constraints is not sufficient for guaranteeing that the generalizer is the least general one. Since a postprocessing that refines the freshness constraints but keeping the generalizer property
does not change the uniqueness.
%
%
 As expected the equational theories will induce  branching when decomposing the specific function symbols, but the possible $\lggeq{E}$'s will be within the finite execution tree of $\enau{E}$, after, possibly, adding some missing  constraints, due to weak completeness.
\begin{theorem}\label{thm:type-aunif}
   $\enau{\emptyset}$ is unitary and $\enau{\tt A}, \enau{\tt C}$ and $\enau{\tt AC}$  are finitary.
\end{theorem}

\begin{proof}
    The format of the rules for $\enau{\emptyset}$ influences the structure of computed generalizers: note that the expression without considering freshness constraints is unique. The algorithm $\enau{\emptyset}$ computes a unique term-generalizer.
\end{proof}
Similarly to~\cite{DBLP:conf/fscd/Schmidt-Schauss22}, completeness can be obtained after a post-processing step using more general forms of constraints that include information of equivalence relations assigning values to atom-variables. A brief discussion and partial extensions are available in \Cref{app:completeness} and is omitted here due to space constraints.

\section{Equivariance Algorithm for \nlat~modulo $E$}\label{sec:equivariance-new}

In this section, we introduce \eqvm, a novel algorithm for solving equivariance problems in \nlat, specifically for handling equivariant (matching) equations of the form \( t \match s \). Unlike the equivariance algorithm developed for the empty theory in \cite{DBLP:conf/fscd/Schmidt-Schauss22}, \eqvm~accounts for the equational theories considered in this work. Handling \(E\) is nontrivial, as atom-variable permutations interact with equational theories in unexpected ways (cf. Example~\ref{ex:equivariance_e}), requiring additional tests and refinements.



\subsection{Preliminaries for \eqvm}
In the merging rule \anewrule{Mer} of $\enau{E}$ we need to decide the {\em equivariance problem modulo $E$}:
Given a set $\Psi$ of equivariant equations of the form $t\match s$
between  \nlat-terms $t$ and $s$,  and a freshness context $\nabla$, find (a representation of) an injective function 
$\pi$ such that $\nabla \models \pi\permef t\approx_E s$, for all $t\match s\in \Psi$.

\begin{theorem}\label{theorem:matching-mod-E-decidable-new}
   The question, whether for two given expressions $s,t$ and freshness constraint $\nabla$, there exists a function 
   $\pi$ such that
$\nabla \models \pi \cdot  s \approx_E t$ is algorithmically decidable. Its complexity has a doubly exponential upper bound.
\end{theorem}

We give a strategy to check equivariance by first splitting the freshness context $\nabla$ into several simple freshness constraints and then
processing them.

\begin{definition}
  A simple freshness context is a set of (simple) freshness constraints of the forms $A\#B$, $A \# X$, or $A \# \lambda B. A$, where $A,B$ are atom-variables and $X$ is a generalization-variable. For simplicity we abbreviate the latter by $A = B$,
\end{definition}

\begin{lemma}\label{lemma:simple-freshness-contexts}
    Every (general) freshness context is equivalent to a disjunction of  simple freshness contexts.
\end{lemma}

\begin{proof}
   Freshness constraints  $A\#t$ where $t$ may be any expression can be replaced by a conjunction of (simpler)  freshness constraints $A\#t_1 \wedge \ldots \wedge A\#t_n$, where the $t_i$ do not contain function symbols. The latter constraints are  of the form $A\# \lambda A_1.\lambda A_2. \ldots \lambda A_n.t_i'$,
   where $t_i'$ is either an atom variable $B$ or a generalization variable $X$. The  first form can be replaced by a disjunction $A = A_1 \vee \ldots \vee A = A_n \vee A\#t$, where $t$ is $B$ or $X$.  Using propositional rules this results in a disjunction of simple freshness contexts.
\end{proof}

\subsection{The  method for checking equivariance of  expressions.}

First, we introduce the decomposition rules of the \eqvm~algorithm are in Figure~\ref{fig:eqvm-simple-freshness-constraints-rules-new}. Second, we present the \eqvm~algorithm (Algorithm~\ref{algo:eqvm}), which incorporates additional steps to compute a solution to our problem.

\subparagraph*{Description of the rules of Figure~\ref{fig:eqvm-simple-freshness-constraints-rules-new}}

\begin{itemize}
    \item Rule $\anewrule{EqAbs}$ applies in equivariant equations between abstractions. That is, to find $\pi$ such that $ \nabla\vDash \pi\cdot ( \lambda W_1.t)\approx_E \lambda W_2.s$. We use the fact that $\nabla \vDash (\lambda W_1.t \approx_E \lambda W_2. s)$ if  $\nabla \vDash (W_1 \ A) \cdot t \approx_E (W_2\ A)\cdot  s$, for a new atom-variable $A$ (i.e. two abstractions are renamings of each other if they can be both renamed using new atom-variable and they continue the same).  Thus, the problem reduces to finding $\pi$ such that   $\nabla \vDash  \pi \cdot (W_1 \ A) \cdot t \approx_E (W_2\ A)\cdot  s$ holds. This corresponds to the problem $(W_1~A)\cdot {t} \match (W_2~A)\cdot {s}$. We add the constraint $A\match A$, as the final $\pi$ shouldn't care about atom-variables introduced during the computation as they do not appear in the input problem.
    \item Rule $\anewrule{Eqf}$ applies in equivariant equations between terms headed with a function symbol, say $f$. To illustrate, let us consider the case with $k=3$. Thus, to find $\pi$ such that $\nabla \vDash \pi\cdot f(t_1,t_2,t_3)\approx_E f(s_1,s_2,s_3)$. Since $\pi$ distributes over $f$, this problem reduces to $\nabla \vDash f( \pi\cdot t_1, \pi\cdot t_2, \pi\cdot t_3)\approx_E f(s_1,s_2,s_3)$, which can be decomposed to three smaller problems $\nabla \vDash \pi\cdot t_i\approx_E s_i$, for $i=1,2,3$. Note that application of $\pi$ does not alter the structure of the term. We can use a similar rule in the case both terms are headed with an associative function symbol $\ffa{f}$ as long as that the terms are flattened w.r.t. $\ffa{f}$.
    \item Rule $\anewrule{EqfC}$ applies in equivariant equations between terms headed with a commutative function symbol $\ffc{f}$. Thus, to find $\pi$ such that $\nabla \vDash \pi\cdot \ffc{f}(t_1,t_2)\approx_E \ffc{f}(s_1,s_2)$. This rule branche due to commutativity, and the problem reduces to find $\pi$ such that $\nabla \vDash \pi\cdot t_1\approx_E s_1$ and $\nabla \vDash \pi\cdot t_2\approx_E s_2$ in one branch, but a different solution could be found in the other branch $\nabla \vDash \pi\cdot t_1\approx_E s_2$ and $\nabla \vDash \pi\cdot t_2\approx_E s_1$. Note that, with nested occurrences of commutative function symbols in a term, this problem has an upper bound of $2^n$ possible solutions, for some $n$.
    \item Rule $\anewrule{EqfAC}$ applies in equivariant equations between terms headed with an associative-commutative function symbol $\ffac{f}$.  The explanation is similar to the one above, except that there is more branching. Here we also assume that the terms are flattened w.r.t. $\ffac{f}$.
\end{itemize}

 The following  example illustrates the interaction between the equivariance problem on atom-variables and the equational theories in $E$.
\begin{example}\label{ex:equivariance_e}
   Consider the equivariance equation $\ffac{f}((B\ C) \cdot A,C,B)\match \ffac{f}(C,A,A)$ and freshness context $\nabla=\{A\#\lambda B.A\}$. The restriction from $\nabla$, that $A=B$ (cf. Example~\ref{ex:abstraction_constraint}), reduces the initial problem to $\ffac{f}((A\ C) \cdot A,C,A)\match \ffac{f}(C,A,A)$ which is equivalent to $\ffac{f}(C,C,A)\match \ffac{f}(C,A,A)$. It is easy to see that $\pi=(A\ C)$ is the unique solution:

   The goal  of \eqvm~is to find a permutation $\pi$ such that $\nabla\vDash\pi\cdot \ffac{f}((B\ C) \cdot A,C,B)\approx_E \ffac{f}(C,A,A)$ is valid, which is equivalent to checking validity of $ \nabla\vDash \ffac{f}((\pi\circ (B\ C)) \cdot A,\pi\cdot C,\pi\cdot B)\approx_E \ffac{f}(C,A,A)$. Note that taking $\pi=(A\ C)$ reduces the problem to checking whether:
   $\nabla\vDash \ffac{f}(((A\ C)\circ (B\ C)) \cdot A,A, (A \ C)\cdot B)\approx_E \ffac{f}(C,A,A)$
   holds for every interpretation of $A,B$ and $C$. Note that $\nabla$ is equivalent to $A=B$ which constraints the interpretation of $A$ and $B$ to always coincide  while allowing unrestricted interpretations of $C$ to concrete atoms.  Say the interpretations of $A,B,C$ are  $a,a$ and $c$, resp. Now the problem reduces to check to the derivability of
   $ \ffac{f}(a,a, c)\approx_E \ffac{f}(c,a,a)$, which is derivable using rule ${\tt (AC)}$ in Figure~\ref{fig:equational_rules}. In more complex examples we might obtain a different permutation in different branches.
\end{example}

\begin{figure}[!t]
    \centering

\hrule
\smallskip

\begin{mathpar}
%
\inferrule*[LEFT=\anewrule{EqAbs}]{\{\abst{W_1}{t} \match \abst{W_2}{s}\}\dotcup \Psi;\nabla  \\
}{
   \{(W_1~A)\cdot {t} \match (W_2~A)\cdot {s}, A\match A\}\dotcup \Psi; \{A \# B, A \# X|~A~\text{is fresh}; B,X \text{ occurring in the input}  \} \cup\nabla;
   }
\\
\and
\inferrule*[LEFT=\anewrule{Eqf}]{ \{f\ttup{t}{k} \match f\ttup{s}{k}\}\dotcup \Psi;\nabla ~~~~\ (\text{same for } \ffa{f})}{
      \cup_{j=1}^k\{t_j\match s_j\}\cup \Psi;\nabla; \Pi }
\and
\inferrule*[LEFT=\anewrule{EqfC}]{\{\ffc{f}(t_1,t_2) \match \ffc{f}(s_1,s_2)\}\dotcup \Psi;\nabla }{
      \{t_1\match s_{i}, t_2 \match s_{3-i}\}\cup \Psi;\nabla}
\and
\inferrule*[LEFT=\anewrule{EqfAC}]{\{\ffac{f}\ttup{t}{n}\match \ffac{f}\ttup{s}{n}\}\dotcup \Psi;\nabla}{ \{t_1\match s_i, \ffac{f}(t_2,\ldots, t_n) \match \ffac{f}(s_1,\ldots, s_{i-1}, s_{i+1}, s_{n})\}\cup \Psi;\nabla}
\end{mathpar}

   \hrule \vspace{1mm}
  \caption{ Decomposition Rules for {\eqvm} with simple freshness context. (n.d.)}
    \label{fig:eqvm-simple-freshness-constraints-rules-new}
\end{figure}

The strategy for checking equivariance involves applying the (n.d.) decomposition to the equivariance equations and then examining the simplified set to construct an injective mapping. This is achieved by systematically testing all possible outcomes of the decomposition. We use Lemma~\ref{lemma:simple-freshness-contexts} to simplify the freshness context of the initial problem.

\begin{myalgorithm}[\eqvm]\label{algo:eqvm} \quad \\
{\bf Input:} $\Psi$ a set of equivariance equations and $\nabla$ a simple freshness context. \\
{\bf Output:} an injective mapping $\Pi$.\\
{\bf 1:} First apply the (non-deterministic, n.d.) algorithm for decomposing in \Cref{fig:eqvm-simple-freshness-constraints-rules-new}.\\
{\bf 2:} Find one result $(S, \nabla)$ of the derivation tree obtained from the application of the rules in \Cref{fig:eqvm-simple-freshness-constraints-rules-new}  that passes the  following construction and tests: \\
{\bf 3:}  Compute  a set of mapping relations 
     $\Pi$ from $S$ as follows: \\
{\bf 4:}   If there is a matching equation $\pi_{i,1}\cdot  X \match  \pi_{i,2}\cdot  Y$ in $S$ then {\bf fail}. \\ 
{\bf 5:}  {\bf While} there is a matching relation between atom-variables $\pi_{i,1}\cdot  A_i  \match   \pi_{i,2} \cdot  B_i$ in $S$ {\bf do} :\\
{\bf 6:}\quad  If all atom variables in $\pi_{i,1}$ are already mapped in the current $\Pi$, but $A_i$ is missing {\bf then}\\
{\bf 7:} \quad define $A_i \mapsto  (\Pi(\pi_{i,1}^{-1}) \circ \pi_{i,2})\cdot B_i$. \\
{\bf 8:}  Else,
          guess a mapping  for the atom-variable
          and repeat. \\
{\bf 9:} {\bf EndWhile}\\
{\bf 10:}    {\bf For each}  matching equation $\pi_{i,1}\cdot  X \match  \pi_{i,2}\cdot  X$ in $S$ {\bf do}:\\
{\bf 11:} Test whether $\nabla \models  \Pi(\pi_{i,1})\cdot X \approx_E \pi_{i,2}\cdot X$.\\
{\bf 12:} ~If there is one pair for which the test fail, return  {\bf fail}.\\
{\bf 13:} ~{\bf EndFor} \hspace{1.5cm} \textcolor{gray}{\% after all pairs are processed, test correctness and injetivity of $\Pi$.}\\
{\bf 15:}  For each pair of  definitions $A_1 \mapsto \pi_{1}\cdot B_1$  and
    $A_2 \mapsto \pi_2\cdot B_2$, the following must hold: \\
{\bf 16:} \qquad
     $\nabla \models  A_1  \approx_E A_2$ if and only if $\nabla \models \pi_1\cdot B_1  \approx_E \pi_2\cdot B_2$.\\
{\bf 17:} In case of success, return $\Pi$ that consists of the defined mappings.

\end{myalgorithm}

Some explanations are in order: In step 4, the algorithm fails if it encounters a problem of the form $\pi_{i,1}\cdot  X \match  \pi_{i,2}\cdot  Y$ as solving it would require finding a $\pi$ such that $\nabla\vDash \pi \cdot (\pi_{i,1}\cdot  X)\approx_E \pi_{i,2}\cdot  Y$. Here, $\pi,\pi_{i,1}$ and $\pi_{i,2}$ are atom-variable permutations, and in the semantics they will be instantiated to concrete atoms, whereas $X,Y$ will be instantiated to ground terms. By Definition~\ref{def:semantics_judgements}, the judgement must hold for all valuations of $X$ and $Y$, and they are independent from each other. The guess in step 8 might be guided by some technique, for e.g., a complete case distinction of equalities between the atom-variables underlying the construction in \cref{theorem:matching-mod-E-decidable-new}. The following examples illustrate the steps of the algorithm~\eqvm.


\begin{example}[Cont. \Cref{ex:first_ac}]\label{ex:first_equiv}
    From the application of \anewrule{Mer} it remains to solve the equivariance problem
    $\eqvm(\{C\match (A\ C)\cdot B, (B\ C)\cdot A\match C\}, \Gamma_1 ), $
    where $\Gamma_1 =\{A\# B,  C\#\lambda A. \ffac{f}(A,A,B), C\#\lambda B. \ffac{f}(A,B,A) \}$ is a subset of $\Gamma'$ without the constraints for $D$ which are irrelevant here:

 \begin{enumerate}
     \item First we need to simplify the context $\Gamma_1$. Note that

    $ C\#\lambda A. \ffac{f}(A,A,B)\iff A=C \vee (A\neq C\wedge C\#B)$

     $ C\#\lambda B. \ffac{f}(A,B,A)\iff C=B \vee (C\neq B\wedge C\#A)$

    Since $A\#B$ is part of the freshness context,
       there is only one possible simpler freshness context since the other cases produce unsatisfiable contexts:  $\Gamma_2 =\{A\#B,A\# C,C\#B\}$, which means that all instantiations of $A,B,C$ are pairwise different.

   \item  We now compute $\eqvm(\{C\match (A\ C)\cdot B, (B\ C)\cdot A\match C\}, \Gamma_2)$:

   No decompositon rule can be applied. In step 5, and focusing on the problem $C\match (A\ C)\cdot B$, we use information from $\Gamma_2$ to simplify it. The `guess' is guided by the following reasoning: Since $A\# B$ and $A\#C$ ensure that  these atoms-variables to have distinct interpretations, we conclude that $(A\ C)\cdot B \equiv B$. This reduces the problem to
   $\{C\match B, (B\ C)\cdot A\match C\}, \Gamma_2$. We repeat step 5, on the equation $(B\ C)\cdot A\match C$, we can use the information in $\Gamma_2$ again, and conclude that $(B\ C)\cdot A\equiv A$, which reduces the problem to $\{C\match  B, A\match C\}; \Gamma_2$.

       The injective mapping for equivariant match is: $\Pi=\{A\mapsto C, C\mapsto B\}$.
 \end{enumerate}
\end{example}

\begin{example}[$\eqvm$ with abstraction and \AC-operator]\label{ex-ac-algo-version}
  Consider the problem with input $\{X:g(\lambda C.\ffac{f}(A,B,C),\lambda A.\ffac{f}(A,B,A))\triangleq g(c, c)\}$ and context $\nabla=\{C \# \lambda A. \lambda B.C\}$.
  Let us check the run of $\enau{E}$:
\begin{enumerate}
\item The first step leads to the following state:\\
  $\{X_1: \lambda C. \ffac{f}(A,B,C) \triangleq c, X_2: \lambda A. \ffac{f}(A,B,A) \triangleq c\}; \emptyset; \nabla;\{X\mapsto g(X_1,X_2)\}$. 

\item  After application of $\arule{Sol}$, we have to try  a  \anewrule{Mer} step, which requires to solve an equivariance match of the expressions:
  $\eqvm(\{\lambda C. \ffac{f}(A,B,C) \match  \lambda A. \ffac{f}(A,B,A)\}, \nabla)$ 

\item We apply the rule \arule{EqAbs} with atom-variable $D$ and obtain:
$$\{\ffac{f}(A,B,D) \match \ffac{f}(D,B,D), D \match D\}; \nabla \cup\{D\#A, D\#B\}$$
In this case, due to the addition of the extra freshness constraints, we do not have to deal with the AC-property of $\ffac{f}$ and obtain $\Pi = \{A \mapsto D\}$. 
\end{enumerate}
\end{example}

\begin{theorem} The following hold:
\begin{enumerate}
\item \eqvm~terminates.
\item \eqvm~is correct.
\end{enumerate}
\end{theorem}

\begin{proof}
    \begin{enumerate}
    \item Since the rules for \eqvm~decompose the problem into smaller problems and each pair $(S,\nabla)$ contains a finite set $S$. The tests in Algorithm~\ref{algo:eqvm} eventually terminate as they consist in checking finite permutations in a finite set of equivariance equations.
    \item \eqvm~is correct iff there exists an injective function $\Pi$ that respects some freshness constraints, say $\nabla'$. Note that $\Pi$ is the representation of a class of injective mappings under $\nabla'$, i.e., given $\Pi=\{A_1\mapsto r_1,\ldots, A_k\mapsto r_k\}$ and $\nabla'\vDash A_i \neq A_j\implies r_i\neq r_j$.

    Let $\{s\match t\}\cup\Psi;\nabla$ be the input to \eqvm, and apply the rules in Figure~\ref{fig:eqvm-simple-freshness-constraints-rules-new} until no more rules can be applied. The result is a finite set of reduced equivariance equations of the form $\pi\cdot A \match \pi'\cdot A$, $\pi\cdot A \match \pi'\cdot B$, $\pi\cdot X \match \pi'\cdot X$, $\pi\cdot X \match \pi'\cdot Y$ and a finite set of simple freshness constraints $\nabla'$, for some permutations $\pi,\pi'$, atom-variables $A,B$ and term-variables $X,Y$. If there is at least one problem of the form $\pi\cdot X \match \pi'\cdot Y$, then no  solution exists, as permutations on atoms do not guarantee equality modulo $E$ on all possible interpretations of different term-variables $X$ and $Y$. For equations of the form $\pi\cdot A \match \pi'\cdot B$ we can always select first the ones for which the atoms in $\pi$ have defined mappings, this will allow to create the mapping associated to $A$. Once those are all checked we can move to equations $\pi\cdot A \match \pi'\cdot B$ for which atoms in $\pi$ are note defined in $\Pi$ yet, but there are restrictions in $\nabla'$ for them. We check the restrictions for each atom, and explore interpretations of the atoms, investigating all possible guesses. After the mappings are defined for all the atomvariables, we check whether $\nabla'\vDash \pi\cdot X \match \pi'\cdot X$ holds. Finally, we check whether $\Pi$ is a representation of an injective mapping of atom-variables.
    \end{enumerate}
\end{proof}

\section{Conclusion and Future Work}
We have investigated the extension of the nominal anti-unification problem to theories \A, \C~and \AC. Our proposed method considers the nominal language with atom-variables \nlat~and a semantic approach to find least general generalizers modulo $E\in\{\tt A,C,AC\}$. The extended algorithm $\enau{E}$ is obtained by adding specific rules to deal with each particular theory as well as new non-deterministic equivariance algorithm \eqvm~that works, semantically, modulo the theories considered here. 
We left as future work the investigation subclasses of the generalization problem that are computationally tractable and lead to a single lgg,  and other optimization methods to work in wider subclasses of the problem, using for e.g., the rigidity functions as in~\cite{DBLP:journals/mscs/CernaK20}. Further equational theories like the considered one with a unit element are also left for future research.



\bibliography{aunif}

\appendix

\section{More Examples}

\begin{example}[Semantics of Abstractions - part 3]
    The semantics of the term-in-context $(\{A\#X\}, \lambda A. f(A,X))$ is the following:
    $$
    \begin{aligned}
    \sem{\{A\#X\}}{\lambda A. f(A,X)}&=\left\{ \eqclass{ \lambda A\rho. f(A\rho,X\rho)} \mid \rho \text{ is an interpretation}\right.
    \left. \hspace{0cm} \text{ and } A\rho \# X\rho \text{ holds }\right\}\\
    &=\{\eqclass{\lambda A\rho. f(A\rho, X\rho)} \mid A\sigma \text{ does not occur free in } X\rho\}\\
    &= \{\eqclass{\lambda a.f(a,b)},\eqclass{\lambda a.f(a,f(c,d))},\eqclass{\lambda b.f(b,\lambda a.a)},\ldots\}
    \end{aligned}
    $$
 the third example of this class is obtained with  $\rho=\{A\mapsto b, X\mapsto \lambda a. a\}$.
\end{example}

\begin{example} When given the input  $P=\{ X: f(A,B) \triangleq f(B,A)\}$ and $\nabla=\emptyset$, the $\enau{\emptyset}$-algorithm computes the generalization $\tc{\emptyset}{f(C_1,C_2)}$, which is an lgg. Note, however, that $\tc{\emptyset}{ f(C,C)}$ is not an lgg for this problem, since it cannot be instantiated back to, for e.g., $f(A,B)$.
 \end{example}

 \begin{example}[Weak completeness] Consider the pairs $(\emptyset, (f(A),A,B))$ and $(\emptyset, (\const{c},A,B))$
    where $\const{c}$ is a constant. The $\enau{\emptyset}$ gives the following derivation:
    \[(\{X: (f(A),A,B) \triangleq (\const{c},A,B)\};\emptyset;\emptyset;Id) \Longrightarrow^{*} (\emptyset;\{Y: f(A)\triangleq \const{c}\};\emptyset; \{X\mapsto (Y, A,B)\})\]
    and the generalization is $(\emptyset, (Y,A,B))$.

   \begin{claim}
   $(\emptyset, (Y,A,B))$ is not an lgg.
   \end{claim}
    Consider the constraint $(\{A\# \lambda B.Y\}, (Y,A,B))$, which is more specific.
    Note that $A\# \lambda B. Y$ restricts the instances of $A,B,Y$ as follows:
    \( (A=B \implies True) \wedge (A\neq B \implies A\# Y)\).

    That is, if $A$ and $B$ are instantiated to the same concrete atom, say $a$, then $A\# \lambda B.Y$ means  $a\# \lambda a.Y$ which is trivially true. On the other hand, if $A$ and $B$ are instantiated to different concrete atoms, say $a,b$, then $A\# \lambda B.Y$ means  $a\# \lambda b.Y$ which simplifies to $a\# Y$.
    By definition,
    \begin{itemize}
    \item $\sem{\emptyset}{ (f(A),A,B))}\subseteq \sem{(\{A\# \lambda B.Y\}, (Y,A,B))}$
    \item $\sem{\emptyset}{ (\const{c},A,B))}\subseteq \sem{\{A\# \lambda B.Y\}, (Y,A,B)}$
        \item $\sem{\{A\# \lambda B.Y\}}{ (Y,A,B)}\subseteq \sem{\emptyset}{ (Y,A,B))}$, i.e., $(\{A\# \lambda B.Y\}, (Y,A,B))$ is more specific. It is not necessary to add constraints on atom-variables that does not occur in the problem. We can prove that this is the lgg.
     \end{itemize}
\end{example}

\begin{example}[Full Computation of~\Cref{ex:first_ac}] 
\label{ex:first_ac2}
    Consider the problem $P=\{ X: \lambda A. \ffac{f}(A,A,B)\triangleq \lambda B. \ffac{f}(A,B,A)\}$ and with $\nabla=\{A\#B\}$ and  the following branch of an $\enau{E}$ derivation:

    \[
    \begin{array}{l}
     P; \emptyset; \nabla;Id\\
     \Lra_{  \arule{Abs}}  \\
    \{Y:  \ffac{f}(C,C,(A\ C)\cdot B)\triangleq \ffac{f}((B\ C)\cdot A,C,(B\ C)\cdot A)\};\emptyset; \Gamma;\{X\mapsto \lambda C. Y\}\\
    \Lra_{\arule{DecAC}}\\
    \left\{
    \begin{aligned}
        &Y_1: C \triangleq (B\ C)\cdot A,\\
        &Y_2: \ffac{f}(C,(A\ C)\cdot B)\triangleq \ffac{f}(C,(B\ C)\cdot A)
    \end{aligned}\right\};\emptyset; \Gamma;\{X\mapsto \lambda C. \ffac{f}(Y_1,Y_2)\}\\
    \Lra_{\arule{DecAC}}\\
    \left\{
    \begin{aligned}
        &Y_1: C \triangleq (B\ C)\cdot A,
        &Y_{21}: C\triangleq (B\ C)\cdot A,\\
        &Y_{22}: (A\ C)\cdot B\triangleq C\\
    \end{aligned}\right\};\emptyset; \Gamma;\{X\mapsto \lambda C. \ffac{f}(Y_1,Y_{21},Y_{22})\}\\

    \Lra_{3\times \arule{SolAB}}\\
     \emptyset; \left\{
    \begin{aligned}
        &D:  C \triangleq (B\ C)\cdot A,
        D':  C \triangleq (B\ C)\cdot A,\\
        &D'': (A\ C)\cdot B\triangleq C
    \end{aligned}\right\}; \Gamma';\{X\mapsto \lambda C. \ffac{f}(D,D',D'')\}\\
\Lra_{1 \times\arule{Mer}}\\
 \emptyset; \left\{
    \begin{aligned}
        D:  C \triangleq (B\ C)\cdot A
    \end{aligned}\right\}; \Gamma'';\{X\mapsto \lambda C. \ffac{f}(D,D,\pi \cdot D)\}\\
    \end{array}
    \]
    where
    \begin{itemize}
    \item $\Gamma =\{A\# B, \allowbreak C\#\lambda A. \ffac{f}(A,A,B), C\#\lambda B. \ffac{f}(A,B,A) \}$,
    \item $\Gamma'\!=\!\Gamma\cup \{D\#\lambda C.\lambda \swap{B}{C}\cdot A.D, \allowbreak D'\#\lambda C.\lambda \swap{B}{C}\cdot A.D', \allowbreak D''\# \lambda \swap{A}{C}\cdot B.\lambda C.D'' \}$
    \item $\Gamma'' =\Gamma \cup \{D\#\lambda C.\lambda (B\ C)\cdot A.D, \; \allowbreak \pi\cdot D\# \lambda (A\ C)\cdot B.\lambda C.\pi\cdot D \}$
    \end{itemize}
    In \Cref{ex:first_equiv} we compute  $\eqvm(\{C\match (A\ C)\cdot B, (B\ C)\cdot A\match C\}, \Gamma' )=\swap{A}{C}\circ \swap{C}{B}$  which we denote as $\pi$ above. Note that $\pi\cdot D$ is a suspension, we only apply the permutation after instantiating $D$ according to the store.  The $E$-generalization obtained is the term-in-context $\tc{\Gamma'}{\lambda C. \ffac{f}(D,D,\pi \cdot D)}$. Note that this is one possible generalization in one branch of the generalization tree. A different branch may produce a different generalization.
\end{example}

\begin{example}[Several non-equivalent lgg-s]
Another variant of the example is $\ffac{f}(s_1,s_2,s_3) \triangleq \ffac{f}(s_1,s_2,s_3,s_4)$, where $s_i$'s are  pairwise different constants, which anti-unifies, for example, to $\ffac{f}(s_1,s_2,X)$, but also to $\ffac{f}(s_1,s_3,X)$,   and $\ffac{f}(s_2,s_3,X)$.
One cannot exploit all three common elements with only one generalization since there is no single $\lggeq{AC}$ that is an instance of all them.
If we would permit the more general theory of AC1, i.e., with a unit, then there is a unique lgg, however, this theory has other properties than $C,A, AC$ insofar as it has a neutral element, and is subject to further research.
\end{example}
We give an   example showing that if  the same constants appear more than once on the same level,
 there may be multiple lgg-s.

\begin{example}[Two different $\lggeq{C}$'s]\label{example:1-C-two-lggs}
   Let $ s\triangleq t$ be a generalization problem with $s$ and $t$ as follows:    %
     $$\begin{array}{lcl}
        s &=& \ffc{f}(\ffc{f}(\ffc{f}(\const{a,a}),\ffc{f}(\const{a,b})),\ffc{f}(\ffc{f}(\const{a,a}),\ffc{f}(\const{b,b})))\\
        t &=& \ffc{f}(\ffc{f}(\ffc{f}(\const{a,a}),\ffc{f}(\const{a,b})),\ffc{f}(\ffc{f}(\const{a,b}),\ffc{f}(\const{a,b})))
    \end{array}
    $$
  Then there are exactly  two non-equivalent (w.r.t. $\approx_C$) $\lggeq{C}$'s  where $X,Y$ are generalization term-variables:
  \begin{enumerate}
      \item  $r_1= \ffc{f}(\ffc{f}(\ffc{f}(\const{a,a}),\ffc{f}(\const{a,b})),\ffc{f}(\ffc{f}(X,\const{a}),\ffc{f}(Y,\const{b})))$
      \item $ r_2=\ffc{f}(\ffc{f}(\ffc{f}(\const{a},X),\ffc{f}(\const{a,b})),\ffc{f}(\ffc{f}(\const{a,a}),\ffc{f}(Y,\const{b})))$
  \end{enumerate}
  In fact, $r_1\not\approx_C r_2$ since $\ffc{f}(\const{a},X)$ occur in  positions in $r_1$ and $r_2$ that  move outside the scope of one $\ffc{f}$ to another, but without respecting commutativity of $\ffc{f}$.
    Note that in a different branch we would get the generalizer $\ffc{f}(\ffc{f}(\ffc{f}(\const{a,a}),\ffc{f}(\const{a,b})),\ffc{f}(\ffc{f}(\const{a},X),\ffc{f}(\const{b},Y)))$   which is the same as the first by applying commutativity of $\ffc{f}$.
\end{example}

\begin{example}
    Consider the problem with terms $s=\ffc{f}(\ffc{f}(\const{a,b}),\ffc{f}(\const{c}, \ffc{f}(\const{a,c})))$ and $t= \ffc{f}(\const{a},\ffc{f}(\ffc{f}(\const{c}, \ffc{f}(\const{a,c})),\const{b}))$. Note that $M_s=\{(\const{a},2), (\const{b},2),(\const{c},2),(\const{a},3),(\const{c},3)\}$,  $M_t=\{(\const{a},1), (\const{c},3),(\const{b},2), (\const{c},4),(\const{a},4)\}$, and $M_{st}=\{(\const{b},2), (\const{c},3)\}$.

    Let $s^*=\ffc{f}(\ffc{f}(*,\const{b}),\ffc{f}(\const{c},*))$ and $t^*= \ffc{f}(*,\ffc{f}(\ffc{f}(\const{c},*),\const{b}))$. Note that $s^* \not\approx_{\tt C}t^*$. There are two $\lggeq{\tt C}$'s   $\ffc{f}(X_1,\ffc{f}(X_2,\const{b}))$ and  $\ffc{f}(X_1',\ffc{f}(\ffc{f}(\const{c}, X_2'), X_3'))$. Note that the former is not more general than the latter.
\end{example}


\section{Proofs of \Cref{ssec:correctnesss}}

\begin{proof}[Proof of \Cref{thm:type-aunif}] Consider the post-processing algorithm for specializing the freshness constraints:
\begin{myalgorithm}
       For the empty theory, the lgg of a generalization problem can be computed as follows:
   \begin{enumerate}
       \item Start with the generalizer $(\nabla,t)$  computed by $\enau{\emptyset}$.
       \item Now use a generate and test method: \\
        For every possible freshness constraint $\nabla'$ containing only atom variables  that occur in the computed solution: Check whether $(\nabla \cup \nabla',t)$  is still a generalizer of the input.
        This check can be performed using  nominal matching of terms-in-context, and consists in verifying whether $(\nabla,t)\subseteq (\nabla \cup \nabla',t)$.
        \item The result is $(\nabla \cup \{\nabla' \mid (\nabla',t)$ is a generalizer of the input\}$, t)$.
   \end{enumerate}
\end{myalgorithm}
\end{proof}

    Testing whether a term-in-context is the same or a subset of another term-in-context is not straightforward.
\begin{example}
   Consider $(\{A\#\lambda C. \lambda D.A\},A)$  and $(\{A'\#\lambda C. \lambda D.A'\},A')$. These have the same semantics:
  $A = C$  or $A = D$ and  $A' = C$  or $A' = D$, respectively.
   Testing this algorithmically cannot be done following the structure alone, but requires a case-analysis.
\end{example}

\section{Regaining completeness: towards an algorithm for post-processing and Subset-Testing of  Terms-in-Context modulo $E$}\label{app:completeness}

In order to obtain an algorithm to compute  lgg-s $(\nabla_g,s_g)$ based on $\enau{}$ a post-processing is added, which takes the resulting generalizers $(\nabla_t,t)$
of $\enau{}$ and then computes the maximal freshness context by a generate-and-test method.

We have a choice between different expressiveness of the freshness constraints.
\begin{enumerate}
    \item\label{fc-item-1} The freshness constraints formed without permutations.
    \item\label{fc-item-2}  The freshness constraints as defined in this paper.
    \item\label{fc-item-3}  The EQR-constraints as in \cite{DBLP:conf/fscd/Schmidt-Schauss22}, which are more general.
\end{enumerate}

\begin{enumerate}
\item Using \cref{fc-item-1} is not sufficiently expressive.
\item Using \cref{fc-item-2} is  general, however, it appears to have a high complexity, since a generate and test-method in the postprocessing
requires an enumeration of  all freshness constraints for a finite set of atom-variables, which presumably is of doubly exponential complexity,
since we have to cover all interpretations.
\item Using \cref{fc-item-3} is  strictly more expressive, and it is not necessary to use permutations in the EQR-constraints, as we will show.
A generate-and-test method is at most exponential, which is better than doubly exponential.
In addition, for small numbers of atom-variables, the number of possible freshness constraints appears to be in the reach of algorithms.
\end{enumerate}

\begin{definition}[EQR-fresheness constraints]
Let  $M$ be a finite set of atom-variables.
An {\em EQR-freshness-constraint} is a pair of the form $(S,C)$, where $S$ is a set of equations and disequations between atom-variables in $M$, and $C$ is a set of  freshness constraints
of the form $A_i \# X$, where $A_i \in M$ and $X$ is a variable, or ${\tt False}$.
As expected, an {\em EQR-freshness-context} is a set of EQR-freshness constraints.

The {\em standardized} form of an EQR-freshness context for a set $M$ of atom-variables is a set $\{(S_i,C_i)~|~i = 1,\ldots,n\}$,
such that  every $S_i$  is an equivalence relation on $M$, and every equivalence relation is there.
\end{definition}

Note that in the non-standardized case we do not require that all sets $S$ are there.

\begin{proposition}\label{prop:EQR-constraints-from-constraints}
    Every (non-EQR) freshness-context can be transformed into an equivalent EQR-freshness-context.
 \end{proposition}

\begin{proof}
Let $\nabla$ be a freshness context. Let $M$ be the set of atom-variables occurring in $\nabla$. Let $E$ be an  equivalence relation on atom-variables of $M$,
and let  $A \# s$  be a  constraint in $\nabla$.
We can assume that $s$ does not contain function symbols, by replacing it by an equivalent set of simpler constraints.
Every
expression $\pi\cdot A\in \nabla$ can be evaluated to an atom-variable by applying the equations/disequations in $E$.
Constraints of the form $A_1 \# \lambda A_2.s$ can be either removed, if $E \models A_1 = A_2$, or simplified to $A_1 \#  s$, otherwise.
Also, all constraints of the form $A_1 \# A_2$ can be removed, if $E \vDash A_1 \not= A_2$. If  $\{A\#A\}$  is in the resulting set,
then the pair is replaced by
$(E,\texttt{False})$.
After applying this, the remaining  simple
constraints are of the form $A  \# X$.
\end{proof}

This implies that the number of freshness constraints that are required in a post-processing test is at most doubly exponential:

\begin{proposition}\label{prop:number-EQR-freshness-constraints}
For a set  $n$ of atom-variables, and $m$ term-variables, an upper bound for the number of EQR-freshness constraints is
$(n*m)^{B_n}$ where $B_n < n^n$ is the Bell-number. Hence this is smaller than $(n*m)^{n^n}$, which is doubly exponential.

\end{proposition}
  \begin{proof} Let us fix $n$, the number of atom variables. Every \nlat-permutation $\pi$ built with $n$ atom-variables maps under an interpretation $\rho$ to a permutation of at most $n$ atoms. The (group-)order of any permutation is a divisor of $n!$, the order of the symmetric group over $n$ elements.
  Hence every interpretation of $\pi^{n!}$ for some NLA-permutation $\pi$ is the unit-element, which is the identity permutation.

  Let's count the possibilities:
  Let $E_1, \ldots, E_k$ be the different equivalence relations of  $\{A_1,\ldots,A_n\}$, where $k$ is their number.  Then $k$ is the Bell-number $B_n$, where $B_n \leq n^n$ is known.
  All possibilities for interpretations of an NLA-permutation $\pi$ are represented by mapping $\pi$ into a tuple $(\pi_1, \ldots,\pi_k)$,
  where $\pi_i$ is the permutation related to $\pi$ using the interpretation with index $i$, resp. the equivalence class $E_i, i = 1,\ldots,k$. This mapping completely determines the semantic behavior of $\pi$ under all interpretations.

  The number of differently interpretable \nlat-permutations is at most $(n^n)^{B_n}$,
   i.e., at most $(n^{n*(n^n)})$, which is doubly exponential.
  It is also not hard to see that this set can also be constructed:
  An enumeration that generates all of them for   1 swapping, 2 swappings, etc, and eliminating redundant ones using the described method will stop when for a number $m$ there are no further NLA-permutations added to the non-redundant set. Then it stops with success. The number of these construction steps is at most $(n^n)^{B_n} = n^{n*B_n}$,
  which implies that their representational length is bounded accordingly.
\end{proof}

\begin{myalgorithm}[Post-processing]
    Given $(\nabla_s,s)$ and $(\nabla_t,t)$ that have to be generalized and a computed generalizer $(\nabla_g,s_g)$,
    the following is iterated, where $M$ is the set of atom-variables in  (where we may assume that the freshness contexts are EQR-constraints)
    \begin{enumerate}
        \item Select an EQR-context $\Delta$.
        \item If $(\nabla_g \cup \Delta,s_g)$ is more general than $(\nabla_s,s)$ and $(\nabla_t,t)$, then let
        $(\nabla'_g,s_g)$ be $(\nabla_g \cup \Delta,s_g)$, where the result of union is represented  as an EQR-constraint.
        Otherwise $(\nabla_g,s_g)$ is unchanged.
        \item Iterate this.
    \end{enumerate}
\end{myalgorithm}

The missing part to be described  is the test for ``being more general'', and as a subalgorithm a consistency test for freshness constraints. Both are given below.
Since only a finite number of  freshness constraints for a fixed finite set of atom- and term-variables has to be tested, this method allows to compute
a unique generalizer of $s,t$ that refines $(\nabla_g,s_g)$.



 \begin{example} We illustrate the idea of the subsumption of two terms-in-context.
 \begin{itemize}
     \item $(\{A\#X\},f(A,X))$ is a proper subset of $(\emptyset,f(A,X))$; it is equal to $(\{B\#Y \},f(B,Y))$;
     and a superset of $(\{A\#X\},f(A,g(X)))$.
     \item $(\emptyset,\lambda A.f(A,X))$  is the same as $(\emptyset,\lambda B.f(B,Y))$.
     \item   $(\{B\#Y\},f(A, g(Y)))$  is not a subset of   $(\{A\#X\},f(A,X))$.
       \item   $(\{A\#Y\},f(A,g(Y)))$  is  a proper subset of   $(\{A\#X,B\#Y\},f(A,X))$.
   \end{itemize}
 \end{example}


In the following we present a sequential procedure, called {\sf TIC-subset}, to test if a Term-In-Context is subset of another, i.e., tests whether $(\nabla_s,s) \subseteq (\nabla_t,t)$.
\begin{myalgorithm}[TIC-subset]
  Let $(\nabla_s,s)$ and $ (\nabla_t,t)$ be two terms-in-context.
\begin{enumerate}
    \item Rename the term-variables and atom-variables of $(\nabla_s,s)$, such that after this operation $(\nabla_s,s)$  is variable disjoint with $(\nabla_t,t)$.
    \item Compute a match $\sigma$  such that $t\sigma  \approx_E s$.
          If this is not possible, then there is no subset relation.
    \item Consider $\nabla_s$  and $(\nabla_t)\sigma$:  Every element of  $\llbracket(\nabla_s,s)\rrbracket$ must be contained in  $\llbracket((\nabla_t)\sigma,t\sigma)\rrbracket$.
       Here we can assume that every freshness constraint is of the form $A\#t$ where $A$ is an atom variable.
       This  condition is satisfied, if  for every constraint $A \# r$  in $(\nabla_t)\sigma$ we have $\nabla_s \vdash A\# r$ (Use \Cref{algo:constraints} to check).
\end{enumerate}
\end{myalgorithm}

Note that steps 2. and 3. check whether $\sigma$ is a solution to an equational matching-in-context problem:   we need to find $\sigma$ such that $\nabla_t\sigma$ holds and $\nabla_t\sigma\models t\sigma\approx_Es$ and $\nabla_t\sigma\models \nabla_s$. The latter is to guarantee that the constraints for $s$ are respected. Algorithms to solve equational nominal matching problems can be found in~\cite{DBLP:conf/mkm/AyalaRinconFSKN23,DBLP:journals/mscs/Ayala-RinconSFS21}.

\begin{myalgorithm}[Test Constraints]\label{algo:constraints}
A simple algorithm for testing whether $\nabla_s \vdash A\# r$ is as follows:
\begin{itemize}
    \item Let $\{a_1,\ldots,a_n\}$ be a set of atoms, where $n$ is the number of atom-variables  in $\nabla_s,A$.
    \item Let $\phi$ be a function mapping atom-variables into the set $\{a_1,\ldots,a_n\}$, and mapping the term-variables to subsets of
    $\{a_1,\ldots,a_n\}$.
    \item Test whether under $\phi$ the relation  $\nabla_s \vdash A\# r$ holds. I.e. either $\phi(\nabla_s)$ is false, or
    $\phi(\nabla_s)$ is true, and then also $\phi(A\# r)$ must hold.
\end{itemize}

\end{myalgorithm}

The problem-class is coNP, hence we can expect only algorithms that are exponential for this problem.

\section{Additional Material: Cases with Unique lggs}\label{sec:optimizations}

The number of computed lggs may in general be very large. Hence, it is useful to investigate optimizations of our algorithm to avoid the computation of redundant ones, and to identify cases where a single lgg exists.
The recognition and efficient computation of a single lgg for a given problem may be useful in practical scenarios where a list $e_1,\ldots,e_n$ of $n$ expressions is given
  and the possibilities for generalization for every pairs $e_i,e_j$ are searched for. A good case is if there is only one or few generalizers for a special pair, which may indicate that the expressions are derived from or even instances of the same general expression.

  We will treat the topic in this section only for some selected special cases, omitting the freshness constraints, and only expressions with a single function symbol (perhaps occurring several times) which may be an \A, \C, or \AC-symbol and the arguments are ground and formed using symbols from $\Sigma_{\emptyset}$, and there are no atoms. It is future work to generalize and extend the criteria to more general situations.

  \subparagraph*{Unique generalizers for \AC.}
 We define a strategy to exhibit unique $\lggeq{\tt AC}$ situations in the case where only a single \AC-function symbol is permitted (possibly multiple occurrences)  and, in addition, signature constants. Thus, we only consider flattened expressions $\ffac{f}(s_1,\ldots,s_n)$ where $n \geq 2$ and $s_i$ are constants.
 More general situations are left for further research.
\begin{example}[Unique $\lggeq{\tt AC}$]\label{example:first-unique-AC}
 Let $\ffac{f}(s_1,s_2,s_3,s_4) \triangleq  \ffac{f}(s_5,s_6,s_1,s_2)$ be a generalization problem where all  $s_i, i=1,\ldots, 6$ are different constants.
  It is sufficient to generalize only $\ffac{f}(s_3,s_4) \triangleq  \ffac{f}(s_5,s_6)$, which results in $\ffac{f}(X_1,X_2)$, and the overall simplified $\lggeq{\tt AC}$  is
of the form
$\ffac{f}(s_1,s_2,X_1,X_2)$.
 In  the slightly more complex case $\ffac{f}(s_1,s_2, s_3,s_4)$  $\triangleq$ $\ffac{f} (s_5,s_6,$ $s_1,s_2,s_7,s_8)$, where all $s_i,s_j$ are pairwise different constants,
 the lgg is also
$\ffac{f}(s_1,s_2,X_1,X_2)$.
\end{example}


  We will use multisets and annotate the corresponding multiset-operations with {\sf ms}.

\begin{myalgorithm}[AC-criterion]\label{def:C-unique-lgg}
   Let $\ffac{f}(s_1,\ldots,s_n) $ and $ \ffac{f}(t_1,\ldots,t_m)$, with $m,n \ge 2$ be the expressions to be generalized:

   \begin{itemize}
       \item If $\ffac{f}(s_1,\ldots,s_n) \approx_{\tt AC}\! \ffac{f}(t_1,\ldots,t_m)$, then there is a unique $\lggeq{AC}$ which is $\ffac{f}(s_1,\ldots,s_n)$.
       \item If  $\ffac{f}(s_1,\ldots,s_n) \not\approx_{\tt AC} \ffac{f}(t_1,\ldots,t_m)$, then we further check the following conditions.
       There are three disjoint multisets $M_1,M_2$ and $ M_3$, such that:
       (i) $M_1 \cup_{\sf ms} M_2 \cup_{\sf ms} M_3 = \{s_1,\ldots,s_n,t_1,\ldots,t_m\}$;
       (ii)  the multisets $M_i$'s are pairwise disjoint and $M_2,M_3 \not= \emptyset$;
       (iii) for $i = 1,2,3$: each $M_i$ does not contain double elements;
       (iii) $M_1 = \{s_1,\ldots,s_n\} \cap_{\sf ms}  \{t_1,\ldots,t_m\}$.
       (iv) $\{s_1,\ldots,s_n\}  = M_1 \cup_{\sf ms} M_2$, and $\{t_1,\ldots,t_m\}  = M_1 \cup_{\sf ms} M_3$.

        Then the unique $\lggeq{\tt AC}$  is $\ffac{f}(r_1,\ldots,r_k,X_1,\ldots,X_l)$,  where $\{r_1,\ldots,r_k\} \approx_E M_1$, and $l$ is the minimum of the cardinality of $M_2$ and $M_3$.
   \end{itemize}
\end{myalgorithm}


\begin{theorem}[\AC-criterion] Let  $X:s\triangleq t$ be a anti-unification problem
 where $s$ and $t$ are as in the \AC-criterion.  \Cref{def:C-unique-lgg} computes the unique $\lggeq{\tt AC}$ for $s$ and $t$ in polynomial time.
\end{theorem}

\subparagraph*{Unique Generalizers for\ \C.} We define a strategy to test a particular unique-$\lggeq{\tt C}$ situation in the case where only a single \C-function symbol, say $\ffc{f}$, is permitted, and also signature constants.
We start with an example for the existence of a unique $\lggeq{\tt C}$:

\begin{example}[Unique $\lggeq{\tt C}$]
    Let $\ffc{f}( \const{a},\ffc{f}(\const{a,b}))$ and $ \ffc{f}(\const{a},\ffc{f}(\const{a,d}))$ be two \nlat~terms to be generalized. There is a unique $\lggeq{\tt C}$, which is $\ffc{f}(\const{a},\ffc{f}(\const{a},X))$.
   The algorithm  $\enau{E}$ will compute four generalizers,  which are all equivalent  to  the $\lggeq{\tt C}$ above,
   thus the computation of one would be sufficient.
    All other generalizers computed in the various branches of the $\enau{E}$ are either equivalent to the $\lggeq{\tt C}$ above (thus redundant, we do not need to compute them multiple times) or equivalent to more general generalizers.
\end{example}


Consider the following set: $M_s := \{(c,i)~|~c \mbox{ is a constant in } s \mbox{ and } i \mbox{ is its depth in } s\}$. Note that the set and number of occurrences of   $(c,i)$ in an expression is invariant w.r.t. $\approx_{\tt C}$.

\begin{myalgorithm}[C-criterion]\label{def:unique-lgg-C}
     Let $s=\ffc{f}(s_1,s_2)$ and $t=\ffc{f}(t_1,t_2)$ be the two expressions to be generalized, where $s_1,s_2,t_1,t_2$ are expressions built from $\ffc{f}$ and signature constants.
    \begin{itemize}
        \item If $\ffc{f}(s_1,s_2) \approx_{\tt C} \ffc{f}(t_1,t_2)$, then there is a unique $\lggeq{\tt C}$.
        \item  Else,  $\ffc{f}(s_1,s_2) \not\approx_{\tt C} \ffc{f}(t_1,t_2)$ and  we apply the following procedure:
        \begin{itemize}
            \item If  $M_s$ and $M_t$ have only unique occurrences of pairs $(c,i)$, then
             let $M_{st} := M_s \cap M_t$.

             Let $s^*$ be derived from $s$ be replacing all maximal subexpressions in $s$ that do not contain subexpressions from  $M_{st}$ by a single element, say $*$. Similarly construct $t^*$. \\
            If  $s^* \approx_{\tt C} t^*$ then
              we compute the $\lggeq{\tt C}$~now from  $s^*$ where all $*$-s are replaced by different variables, resulting in the $\lggeq{\tt C}$-expression.
            \item Else, more than one $\lggeq{\tt C}$ might exist.
         \end{itemize}
       \end{itemize}
  \end{myalgorithm}


\begin{theorem}[\C-criterion] The criterion in \Cref{def:unique-lgg-C}  guarantees a unique $\lggeq{\tt C}$
    of the two input expressions with a top-\C-symbol. The computation is in polynomial time.
\end{theorem}

\begin{example}
    Consider the problem with $s=\ffc{f}(\const{a},\ffc{f}(\const{a,b}))$ and $t= \ffc{f}(\const{b},\ffc{f}(\const{a,d}))$. By definition,  $M_s=\{(\const{a},1), (\const{a},2),(\const{b},2)\}$,  $M_t=\{(\const{b},1), (\const{a},2),(\const{d},2)\}$, and $M_{st}=\{(\const{a},2)\}$. Let $s^*=\ffc{f}(*,\ffc{f}(\const{a},*))\approx_{\tt C} t^*$. The $\lggeq{\tt C}$ is  $\ffc{f}(X_1,\ffc{f}(\const{a},X_2))$.
\end{example}


\subparagraph*{Unique Generalizers for \A.}

  We define a criterion for unique-$\lggeq{\tt A}$ situations where only  a single \A-function symbol is permitted (possibly multiple occurrences), and signature constants.

\begin{myalgorithm}[\A-criterion]\label{def:lgg-for-A}

   Let  $s = \ffa{f}(s_1,\ldots,s_n)$ and $t = \ffa{f}(t_1,\ldots,t_n)$ be the two expressions to be generalized,
   where $s_i,t_i$ are signature constants for all $i$.
   \begin{enumerate}

     \item The string $s_1\ldots s_n$ is of the form $Q_1 S_1 Q_2 S_2\ldots S_{n-1} Q_n$, where  $Q_i, S_i$ are non-empty substrings of constants with the possible exceptions of $Q_1,Q_n$.
       \item The string $t_1\ldots t_n$ is of the form $Q_1T_1Q_2T_2\ldots T_{n-1}Q_n$, where  $T_i$ are non-empty substrings of constants for all $i$.
       \item There are three disjoint sets $M_Q,M_S, M_T$ of signature constants, such that $Q_i$ only contains $M_Q$-symbols, $S_i$ only $M_S$-symbols and $T_i$ only $M_T$-symbols.
       \item The strings $Q_1\ldots Q_n$, $S_1\ldots S_{n-1}$ and $T_1\ldots T_{n-1}$ are linear, i.e., there are no double occurrences of symbols in them.
   \end{enumerate}
   Then the unique $\lggeq{\tt A}$  is $s_1 X_{1,1}\ldots X_{1,k_1}s_2\ldots s_{n-1}X_{n-1,1} \ldots X_{n-1},\ldots X_{n-1,k_n}s_n$, where $k_j$ is the minimum of the lengths of $S_j$ and $T_j$, and $X_{i,j}$ are generalization variables.
\end{myalgorithm}

\begin{example}
Let $s = ab c_1c_2c_3d$  and $t = abc_4c_5d$. Then an $\lggeq{\tt A}$ is  $r = abY_1Y_2d$. Note that for $\sigma_1=\{Y_1\mapsto c_1, Y_2\mapsto c_2c_3\}$ and $\sigma_2=\{Y_1\mapsto c_4, Y_2\mapsto c_5\}$, we get $r\sigma_1=s$ and $r\sigma_2=t$.

\end{example}

\begin{theorem}[\A-criterion]   Let  $X:s\triangleq t$ be a anti-unification problem
 where $s$ and $t$ are as in the \A-criterion.
    Then the  computed generalizer in \Cref{def:lgg-for-A} is a unique $\lggeq{\tt A}$ and it is computed  in polynomial time.
\end{theorem}

\end{document}